\documentclass[10pt]{article}
\usepackage[paper=letterpaper,centering,margin=1.0in]{geometry}

\usepackage{setspace}
\usepackage[american]{babel}
\usepackage{epsfig}
\usepackage{color}
\usepackage{amsthm,amsfonts,amsmath,amssymb,amstext,latexsym}
\allowdisplaybreaks

\newtheorem{theorem}{Theorem}
\newtheorem{lemma}{Lemma}

\newtheorem{fact}{Fact}

\newtheorem{case}{Case}
\newtheorem{problem}{Problem}
\newtheorem{example}{Example}

\newcommand{\Var}{\mbox{Var}}
\newcommand{\A}{\mathcal{A}}
\newcommand{\J}{\mathcal{C}}

\newcommand{\E}{\mathbb{E}}

\newcommand{\bydef}{\stackrel{\rm{def}}{=}}

\title{Control of parallel non-observable queues: asymptotic equivalence and optimality of periodic policies\footnote{Research partially supported by grant SA-2012/00331 of the Department of Industry, Innovation, Trade and Tourism (Basque Government) and grant MTM2010-17405 (Ministerio de Ciencia e Innovaci\'on, Spain) which sponsored an internship of T. Nesti at BCAM.}
}

\author{
Jonatha Anselmi\footnote{Basque Center for Applied Mathematics (BCAM), Al. de Mazarredo 14, 48009 Bilbao, Basque Country, Spain; INRIA Bordeaux Sud Ouest, 200 av. de la Vieille Tour, 33405 Talence,  France. Email: jonatha.anselmi@inria.fr},
~Bruno Gaujal\footnote{INRIA and LIG, Zirst 51, Av. J. Kuntzmann, MontBonnot Saint-Martin, 38330, France. Email: bruno.gaujal@inria.fr}
~and Tommaso Nesti\footnote{Basque Center for Applied Mathematics (BCAM), Al. de Mazarredo 14, 48009 Bilbao, Basque Country, Spain; Dipartimento di Matematica, University of Pisa, Largo B. Pontecorvo 5, 56127 Pisa, Italy. Email: nesti@mail.dm.unipi.it}
}

\begin{document}
\date{}
\maketitle

\begin{abstract}
We consider a queueing system composed of a dispatcher that routes jobs to a set of non-observable queues working in parallel.
In this setting, the fundamental problem is which policy  should the dispatcher implement to minimize the stationary mean waiting time of the incoming jobs.
We present a structural property that holds in the classic scaling of the system where the network demand (arrival rate of jobs) grows proportionally with the number of queues.
Assuming that each queue of type~$r$ is replicated~$k$ times,
we consider a set of policies that
are periodic with period $k \sum_r p_r$ and such that exactly $p_r$ jobs are sent in a period to each queue of type~$r$.
When $k\to\infty$, our main result shows that all the policies in this set are \emph{equivalent}, in the sense that they yield the same mean stationary waiting time, and \emph{optimal},
in the sense that no other policy having the same aggregate arrival rate to \emph{all} queues of a given type can do better in minimizing the stationary mean waiting time.
This property holds in a strong probabilistic sense.
Furthermore, the limiting mean waiting time achieved by our policies is a convex function of the arrival rate in each queue, which
facilitates the development of a further optimization aimed at solving the fundamental problem above for large systems.
\end{abstract}

\section{Introduction}

In computer and communication networks,
the access of jobs to resources (web servers, network links, etc.) is usually regulated by a dispatcher.
% , which operates to optimize some performance criterion.
A fundamental problem is which algorithm should the dispatcher implement to minimize the mean delay experienced by jobs.
There is a vast literature on this subject
and
% it turns out that 
the structure of the optimal algorithm strongly depends on
\emph{i)} the information available to the dispatcher, 
\emph{ii)} the topology of the network and
\emph{iii)} how jobs are processed by resources.
We are interested in a scenario where:
\begin{itemize}
\item The dispatcher has \emph{static} information of the system;

% \item The dispatcher routes jobs among a set of parallel resources;
\item The network topology is parallel;

\item Resources process jobs according to the first-come-first-served discipline.
\end{itemize}
Static information means that the dispatcher knows the probability distributions of job sizes and inter-arrival times 
but cannot observe the dynamic state of resources such as the current number of jobs in their queues.
% This is motivated by the fact that real networks are composed of hundreds or thousands of resources,
% and allowing for dynamic information implies a non-negligible communication overhead and the problem of delayed information \cite{330512}.
% On the other hand, inferring statistics of inter-arrival times and job sizes in advance is a much easier task.
% The scenario above is of interest in volunteer computing, cloud computing, web server farms, etc.; see, e.g., \cite{JavadiKVA11,JavadiTB12,HBSY09} respectively.
This scenario can be of interest in the context of volunteer computing, cloud computing, web server farms, etc.; see, e.g., \cite{JavadiKVA11,JavadiTB12,HBSY09} respectively.

In this framework, the problems
of finding an algorithm, or \emph{policy}, that minimizes the mean stationary delay and
of determining the minimum mean stationary delay
are both considered difficult; see, e.g., \cite{Gau2000,Gau2000acm} for an overview.
A policy can be defined as a function that maps a natural number $n$, corresponding to the $n$-th job arriving to the dispatcher, to a probability mass function $P_n$ over the set of resources.
When the $n$-th job arrives, the dispatcher sends it to resource $i$ with probability $P_n(i)$.
% Under Markovian assumptions, an optimal policy can be computed using the framework of Markov decision processes.
% However, due to the curse of dimensionality of the problem, the computational complexity of this approach is prohibitively expensive even for a small number of resources, e.g., three.
Unfortunately, the problem of finding an optimal policy is intractable and for this reason two extreme families of policies have received particular attention in the literature:
\emph{probabilistic} policies, obtained when $P_n$ is constant (in $n$), and
\emph{deterministic} policies, obtained when $P_n$ puts the whole mass on a single resource.

When dealing with probabilistic policies, 
the difficulty of the problem is simplified by the fact that the arrival process at each resource is a renewal process, provided that the same holds for the arrival process at the dispatcher.
This allows one to decompose the problem and, using the theory of the mean waiting time of the single GI/GI/1 queue,
to immediately reduce it to a relatively simple optimization problem.
This problem is usually convex and there exist efficient numerical procedures for their solution; e.g.,
\cite{NeelyM05,SX97,Squillante1999,Borst1995,ComBox95,BS83}.

Contrariwise, when dealing with deterministic policies,  one of the main difficulties is that the arrival process at each resource is hardly ever a renewal process.
This prevents one from decomposing the problem and directly using the classic theory of the single queue as it has been done for probabilistic policies.
Given this difficulty, researchers divided this problem in two subproblems:
\begin{itemize}
 \item[\emph{i)}] In the first subproblem, the optimal deterministic policy is searched among all the deterministic policies ensuring that the long-term fractions of jobs to be sent to each resource is kept fixed (denote such fractions by vector $p$);
 \item[\emph{ii)}] In the second subproblem, the output of the first subproblem is employed to develop a further optimization over $p$.
\end{itemize}
In this paper, we focus on the first subproblem and, under some system scaling, we identify a set of policies that are optimal.
This result is used to reduce the second subproblem to the solution of a convex optimization.

One of the folk theorem of queueing theory says that determinism in the inter-arrival times minimizes the waiting time of the single queue \cite{Hajek1983,humblet1982determinism}.
In view of this classic insight and fixing fractions $p$,
it is not surprising that an optimal policy tries to make the arrival process at each resource as
\emph{regular}
% \footnote{The concept of regularity for arrival sequences is related to the concept of Schur convexity; see \cite{Altman2002} for further details.}
(or less variable) as possible.
Thus, our stochastic scheduling problem can be essentially converted into a problem in word combinatorics.
% One of the main results known is that
% if $p$ is \emph{balanceable}, then the structure of an optimal policy is known \cite[Theorem~3.3]{Gau2000acm}.
% 
If the dispatcher must ensure fractions~$p$,
the main result known in the literature is that \emph{balanced sequences} are optimal admission sequences \cite{Haj86,Gau2000acm}.
However, balanced sequences of given rates $p$ are known to exist in very few particular cases.
These cases are captured by Fraenkel's conjecture, which is still open to the best of our knowledge \cite{Tijdeman00}; see also \cite[Chapter~2]{Gau2000}, which contains an overview of which rates $p$ are balanceable.

Matter of fact, the problem of finding an optimal deterministic policy is still considered difficult \cite{BhulaiFHL12,AnselmiG10,Laan05,hor04,Gau2000,BNS98,Hordijk94analysisof}.
The only exceptions are when resources are stochastically equivalent, where round-robin\footnote{Round-robin sends the $n$-th job to resource $(n\bmod
R)+1$, where $R$ is the total number of resources.} is known to be optimal in a strong sense \cite{LR98}, or 
when the dispatcher routes jobs to two resources, where balanced sequences can be always constructed no matter the value of rates $p$ \cite{Haj86}. In presence of more than two queues, we stress that balanced sequences with given rates $p$ do not exist in general.
This non-existence makes the problem difficult and one still wonders which structure should an optimal policy have when~$p$ is not balanceable.
When the routing is performed to two resources, jobs join the dispatcher following a Poisson process and service times have an exponential distribution,
the optimal rates $p$ as function of the inter-arrival and service times have a fractal structure, see \cite[Figure~8]{GaujalHyon}. This puts further light on the complexity of the problem even in a simple scenario.

While deterministic policies are believed to be more difficult to study than probabilistic policies (deterministic and probabilistic in the sense described above),
they can achieve a significantly better performance~\cite{AnselmiG10}.
This holds also for the variance of the waiting time because, as discussed above, the arrival process at each resource is much more regular in the deterministic case, especially if there are several
resources as we show in this paper.
% (this will be also commented in Section~\ref{eq:discussion}).
% 
% 
% To the best of our knowledge, examples of other related open questions are
% \emph{i)} whether there exists an optimal policy (in the sense described above) that is deterministic, 
% \emph{ii)} whether optimal policies have rates, i.e., the limit as $t\to\infty$ of the fractions of jobs sent to each resource by time~$t$.
% 
% Good numerical results have been also found when the dispatcher routes jobs to resources following particular billiard policies~\cite{AnselmiG10},
% which have been used in a volunteer-computing system~\cite{JavadiKVA11}.
A particular class of deterministic policies, namely \emph{billiard} policies, have been recently implemented in the context of large volunteer and cloud computing to improve the performance of real applications such as SETI@home~\cite{JavadiKVA11,JavadiTB12}.

\subsection{Contribution}

In the framework described above, we are interested in deriving structural properties of deterministic policies
when the system size is large.
We study a scaling of the system where the arrival rate of jobs, $\lambda k$, grows to infinity proportionally with the number of resources (queues in the following), $R k$,
while keeping the network load (or utilization) fixed. This scaling is often used in the queueing community.
Specifically, there are $R$ types of queues, and $k$ is the number of queues in each type, i.e., the parameter that we will let grow to infinity.
Beyond issues related to the tractability of the problem, this type of scaling
is motivated by the fact that the size of real systems is large and
that replication of resources is commonly used to increase system reliability.

First, with respect to a class of periodic policies, we define the random variable of the waiting time of each incoming job.
This is done using Lindley's equation \cite{lindley} and a suitable initial randomization.
Using such randomization, we can adapt the framework developed by Loynes in \cite{Loynes} to our setting where jobs are sent to a set of parallel queues.
In particular, Theorem~\ref{th0} shows the monotone convergence in distribution of the waiting time of each incoming job.
% XXX
% To the best of our knowledge, this is the first paper that characterizes the stationary distribution of the waiting time of incoming jobs in parallel FCFS queues when policies are deterministic (see more details in Section~\ref{eq:discussion}). 
% 
Then, with respect to a given vector $p\in\mathbb{N}^R$,
we define a certain subset of policies that
are periodic with period $k \sum_r p_r$ and such that exactly~$p_r$ jobs are sent in a period to each queue of type~$r$.
While further details will be developed in Section~\ref{sec:per_pol}, this set is meant to imply that queues of a given type are visited in a round-robin manner and that arrivals are ``well distributed'' among the different queue types.
When $k\to\infty$, our main result states that all the policies in this set are \emph{equivalent},
in the sense that they yield the same mean stationary waiting time, and \emph{optimal},
in the sense that no other policy having the same aggregate arrival rate to \emph{all} queues of a given type can do better in minimizing the mean stationary waiting time.
In particular,  we show that the stationary waiting time converges both in distribution and in expectation
to the stationary waiting time of a system of independent D/GI/1 queues
whose parameters only depend on $p$, $\lambda$ and the distribution of the service times.
This is shown in Theorems~\ref{th05} and~\ref{th1}, respectively.

The main idea underlying our proof stands in analyzing the sequence of stationary waiting times along appropriate subsequences.
Along these subsequences, it is possible to extract a pattern for the arrival process of each queue that is common to all members of the subsequence.
Such pattern is exploited to establish monotonicity properties in the language of stochastic orderings. 
These properties hold for the considered subsequences only: they do not hold true along any arbitrary subsequence and counterexamples can be given.
These properties will imply the uniform integrability of the sequence of stationary waiting times and will allow us to work on expected values.

Summarizing, fixing the proportions~$p$ of jobs to send to each queue type and given $k$ large,
our results state that all the policies belonging to the set that we identify in Section~\ref{sec:per_pol} yield the same asymptotic performance and are asymptotically optimal.
Furthermore, using known properties of the D/GI/1 queue, we obtain that the stationary mean waiting time obtained in our limit is a convex function of~$p$.
This reduces the complexity of subproblem \emph{ii)} above because it boils down to the solution of a convex optimization problem.
%  as we briefly discuss in the Conclusions

% We have performed simulations to observe how fast expected values and variances converge to their asymptotic value and
% to test our results against a real arrival process for web servers.
% Surprisingly, it turns out that the mean stationary waiting time is very
% close to its asymptotic value even when $k=10$.
% A similar observation holds also for the variance of the mean stationary waiting time.
% Then, we have compared the performance of the worst policy in our set with the performance of billiard policies, which are known to have very good performance \cite{hor04}. Both policies belong to our
% set.
% We have performed a trace-driven simulation with data taken from real measurements of web servers and publicly available at the Internet Traffic Archive.
% Surprisingly, we show that the mean waiting times achieved by both of them are very close each other even when $k$ is small, i.e., 10.

This paper is organized as follows.
Section~\ref{sec:model} introduces the model under investigation and provides a characterization of the stationary waiting time (Theorem~\ref{th0});
Section~\ref{sec:mainres} introduces a class of policies and presents our main results (Theorems~\ref{th05} and~\ref{th1});
Section~\ref{sec:proofs} is devoted to proofs;
% Section~\ref{sec:numx} presents simulation results;
finally, Section~\ref{sec:concl} draws the conclusions of this paper.

\section{Parallel queueing model}
\label{sec:model}

% This section presents and justifies all the assumptions that will be used in this paper.

We consider a queueing system composed of~$R$ types of {queues} (or resources, servers) working in parallel.
Each queue of type $r$ is replicated $k$ times, for all $r=1,\ldots,R$, so there are $kR$ queues in total.
Parameter $k$ is a scaling factor and we will let it grow to infinity.
The service discipline of each queue is first-come-first-served (FCFS) and the buffer size of each queue is infinite.
A stream of {jobs} (or customers) joins the queues through a {dispatcher}. 
The dispatcher routes each incoming job to a queue according to some policy and instantaneously.
Figure~\ref{fig:model} illustrates the structure of the queueing model under investigation.
\begin{figure}
\centering
\includegraphics[width=10cm,height=6.5cm]{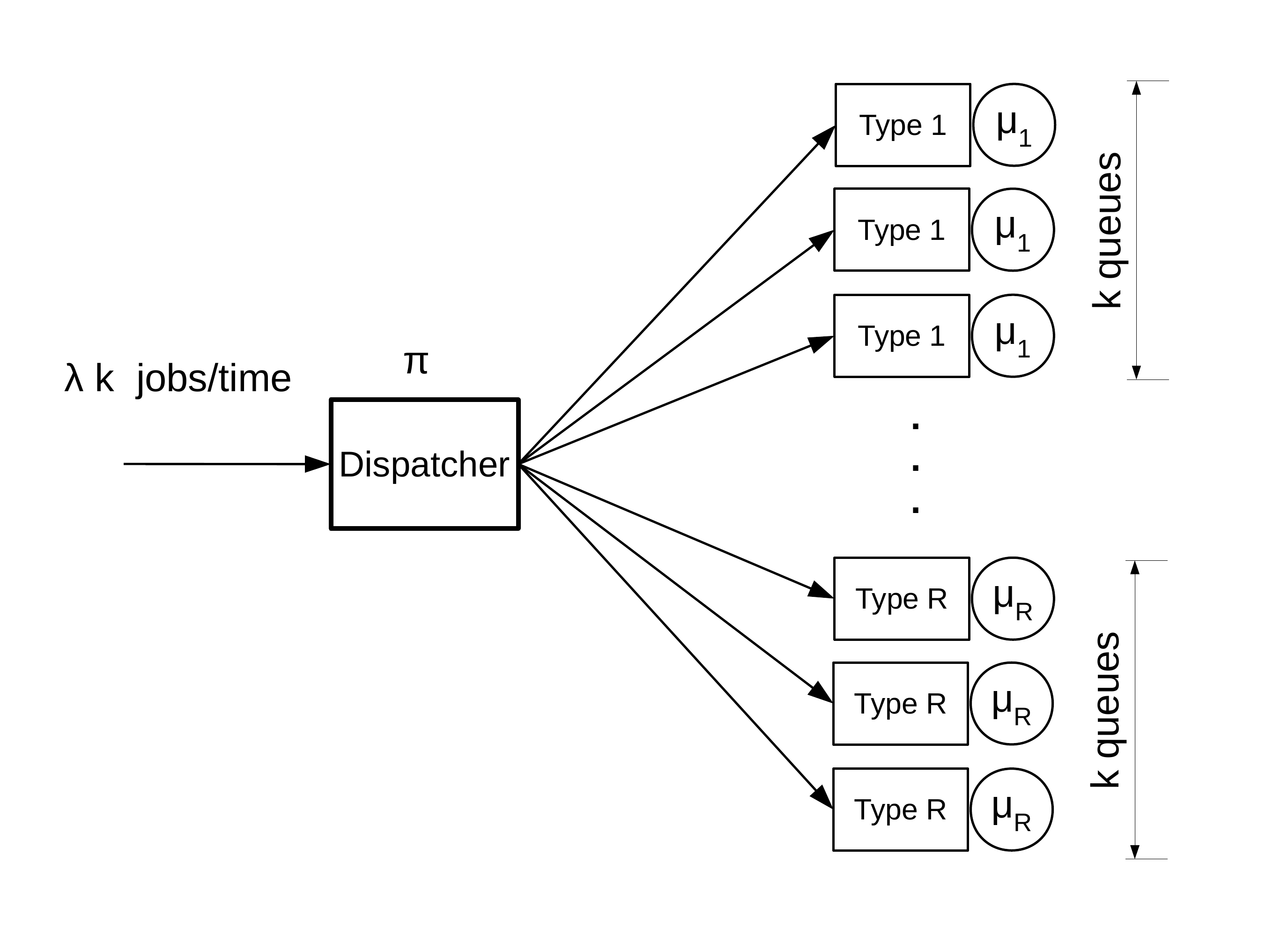}
\caption{Structure of the parallel queueing model under investigation.}
\label{fig:model}
\end{figure}
In the following, indices $r$, $\kappa$, $n$ will be implicitely assumed to range from 1 to $R$, from 1 to $k$, in $\mathbb{N}$, respectively.

All the random variables that follow will be considered belonging to a fixed underlying probability triple~$(\Omega,\mathcal{F},\Pr)$.

Let $(T_{n}^{(k)})_{n\in\mathbb{N}}$ and $(S_{n,\kappa,r}^{(k)})_{n\in\mathbb{N}}$ be given sequences of i.i.d. random variables in $\mathbb{R}_+$\footnote{For any $E\subseteq \mathbb{R}$, we let $E_+\bydef \{x\in E:x>0\}$.}.
These sequences are all assumed to be independent each other.
Quantity $T_{n}^{(k)}$ is interpreted as the inter-arrival time between the $n$-th and the $(n+1)$-th jobs arriving to the dispatcher.
Quantity $S_{n,\kappa,r}^{(k)}$ is interpreted as the service times of the $n$-th job arriving at the $\kappa$-th queue of type~$r$.
% Fixed $r$, we assume that the processes $(S_{n,\kappa,r}^{(k)})_{n\in\mathbb{N}}$, for all $\kappa$, are independent replica of each other and that $\E S_{n,\kappa,r}^{(k)} =\mu_r^{-1}$.
We assume that $S_{1,\kappa,r}^{(k)}=_{st} S_{1,1,r}^{(k)}$ and that $\E S_{n,\kappa,r}^{(k)} =\mu_r^{-1}$.
For the arrival process at the dispatcher, we will refer to the following cases.
\begin{case}
\label{as0}
The process $(T_{n}^{(k)})_{n\in\mathbb{N}}$ is a renewal process with rate $\lambda k$ and such that \emph{Var}$\, T_{n}^{(k)} = o(1/k)$.
\end{case}

\begin{case}
\label{as1}
The process $(T_{n}^{(k)})_{n\in\mathbb{N}}$ is a Poisson process with rate $\lambda k$.
\end{case}
\begin{case}
\label{as2}
The process $(T_{n}^{(k)})_{n\in\mathbb{N}}$ is constant with rate $\lambda k$, i.e., $T_{n}^{(k)}=(\lambda k)^{-1}$.
\end{case}
It is clear that Cases~\ref{as1} and~\ref{as2} are both more restrictive than Case~\ref{as0}.

Let $\|\cdot\|$ denote the $L_1$-norm.

Let $q\bydef (q_{r,\kappa})\in \mathbb{Q}_+^{R}\times \mathbb{Q}_+^{k}$ be
such that $\|q\|=1$. Quantity $q_{r,\kappa}$ will be interpreted as the proportion of jobs sent to queue $(r,\kappa)$.

Let $n^*\bydef n^*(q)\bydef \min\{n\in\mathbb{Z}_+: n\,q \in\mathbb{Z}_+^R\times \mathbb{Z}_+^k\}$.
Since $q$ is a vector of rational numbers, $n^*<\infty$.

% Let $V$ be a discrete random variable with values in $\{1,\ldots,R\}\times \{1,\ldots,k\}$ such that  $\Pr(V=(r,\kappa))=1/(R k)$.
Let $V$ be a discrete random variable with values in $\{1,\ldots,n^*\}$ such that 
$\Pr(V=i)=1/n^*$, for all $i=1,\ldots,n^*$.
We assume that $V$ is independent of any other random variable.

% Let $\A_q(k)$ be the set of all functions $\pi: \mathbb{N} \to \{1,\ldots,R\}\times \{1,\ldots,k\}$ such that
% \begin{equation}
% q_{r,\kappa} = \lim_{n\to\infty} \frac{1}{n} \sum_{n'=1}^n \mathbf{1}_{\{\pi(n')=(r,\kappa)\}}, \quad \forall r,\kappa
% \end{equation}
% where $\mathbf{1}_{E}$ denotes the indicator function of event~$E$.
Let $\A_q(k)$ be the set of all functions $\pi: \mathbb{N} \to \{1,\ldots,R\}\times \{1,\ldots,k\}$ such that
for all $r$ and $\kappa$
\begin{equation}
% q_{r,\kappa} = \lim_{n\to\infty} \frac{1}{n} \sum_{n'=1}^n \mathbf{1}_{\{\pi(n')=(r,\kappa)\}}, \quad \forall r,\kappa
% q_{r,\kappa} = \frac{1}{n^*} \sum_{n=1}^{n^*} \mathbf{1}_{\{\pi(n+n')=(r,\kappa)\}}, \quad \forall n'>0
% 
q_{r,\kappa} = \frac{1}{n^*} \sum_{n=V}^{n^*+V-1} \mathbf{1}_{\{\pi(n)=(r,\kappa)\}} \mbox{ \,and\, } \pi(n)=\pi(n+n^*)
\end{equation}
for all $n$, where $\mathbf{1}_{E}$ denotes the indicator function of event~$E$.
Thus, these functions are periodic with period $n^*$ and $n^* q_{r,\kappa}$ is the number of jobs sent in a period to queue~$(r,\kappa)$.
We refer to each element $\pi\in\A_q(k)$ as a \emph{policy} (or a $q$-policy) operated by the dispatcher, and it is interpreted as follows:
$\pi(n)=(r,\kappa)$ means that the $n$-th job arriving to the dispatcher is sent to the $k$-th queue of type $r$ if
% $n\ge \min\{n':\pi(n')= V\}$,
$n\ge V$,
otherwise it means that the $n$-th job is discarded.
% Thus, the outcome of random variable $V$ represents the first queue that serves a job.
Thus, the outcome of random variable $V$ gives the index of the first job that is actually served by some queue.
In other words, $\pi(V)$ is the first queue that serves some job.

Let $(T_{n,\kappa,r}^{(k)}(\pi))_{n\in\mathbb{N}}$ be the sequence of inter-arrival times that are induced by policy~$\pi$ at the $\kappa$-th queue of type $r$ (under any of the cases above).
By construction, $T_{n,\kappa,r}^{(k)}(\pi)$ is the sum of a deterministic number of inter-arrival times seen at the dispatcher.
The arrival process $(T_{n,\kappa,r}^{(k)}(\pi))_{n\in\mathbb{N}}$ can be made stationary if it is allowed a shift in time and a suitable randomization (independent of $V$ and of any other random variable) for the inter-arrival time of the first arrival of each queue.
We assume for now that this has been done (details will be given at the beginning of
Section~\ref{sec:proofs}).
As done in \cite{Loynes} and according to \cite[p. 456]{Doob}, this implies that we can extend the stationary process
$(T_{n,\kappa,r}^{(k)}(\pi))_{n\in\mathbb{N}}$ to form a stationary process $(T_{n,\kappa,r}^{(k)}(\pi))_{n\in\mathbb{Z}}$
(clearly, the same holds for the process $(S_{n,\kappa,r}^{(k)})_{n\in\mathbb{N}}$).

The waiting time of the $n$-th job arriving to the $\kappa$-th queue of type $r$ induced by a policy $\pi\in \A(k)$ is denoted by
$W_{n,r,\kappa}^{(k)}(\pi)$.
It is the time between its arrival at the dispatcher (or equivalently at the
queue) and the start of its service, and it is defined as follows:
for $n=0$, $W_{n,r,\kappa}^{(k)}(\pi)=0$ and for $n>0$,
\begin{equation}
\label{eq:lindley_xx}
% W_{n,r,\kappa}^{(k)}(\pi) = w_n\left( \left(S_{n',\kappa,r}^{(k)}-T_{n',\kappa,r}^{(k)}(\pi)\right)_{n'=1,\ldots,n} \right)
W_{n,r,\kappa}^{(k)}(\pi) \bydef \left(W_{n-1,r,\kappa}^{(k)}(\pi) + S_{n,\kappa,r}^{(k)}-T_{n,\kappa,r}^{(k)}(\pi)\right)^+,
\end{equation}
where $x^+\bydef \max\{x,0\}$.
Equation \eqref{eq:lindley_xx} is known as Lindley's recursion \cite{lindley}.
The assumption that $W_{0,r,\kappa}^{(k)}(\pi)=0$ serves to avoid technicalities\footnote{Using a standard coupling argument and that each queue will empty in finite time almost surely, what follows can be generalized easily to the case where $W_{0,r,\kappa}^{(k)}(\pi)\ge 0$.}.
It is known that the sequence of random variables $(W_{n,r,\kappa}^{(k)}(\pi))_{n\in\mathbb{N}}$ converges in distribution to the random variable
\begin{equation}
\label{eq:lim_W_rk}
W_{r,\kappa}^{(k)}(\pi) \bydef \left(\sup_{n\ge 0} \sum_{n'=0}^n S_{-n',\kappa,r}^{(k)}-T_{-n',\kappa,r}^{(k)}(\pi)\right)^+
\end{equation}
and that $W_{n,r,\kappa}^{(k)}(\pi)\le_{st}W_{n+1,r,\kappa}^{(k)}(\pi)$, where $\le_{st}$ denote the usual stochastic order; see \cite{Loynes}.
% Thus, also the moments of $W_{n,r,\kappa}^{(k)}(\pi)$ converge to the moment of \eqref{eq:lim_W_rk}.
% We observe here that the $W_{r,\kappa}^{(k)}(\pi)$'s live in the same probability space of the $W_{n,r,\kappa}^{(k)}(\pi)$'s.
We refer to $W_{r,\kappa}^{(k)}(\pi)$ as the \emph{stationary} waiting time of jobs at the $\kappa$-th queue of type $r$.

Given $\pi$, let $f:\mathbb{N}\to\mathbb{N}\times \{1,\ldots,R\}\times \{1,\ldots,k\}$ be a mapping with the following meaning:
$f(n)=(n',r,\kappa)$ means that the $n$-th job arriving to the dispatcher is the $n'$-th customer joining queue $(r,\kappa)$.
If $n<V$, no job is sent to any queue and thus we assume $f(n)=0$.
% Note that the source of randomness in $f(n)$ is only in random variable~$V$.
Note that $f(n)$ is a deterministic function of random variable $V$.
Since $V$ is uniform over $\{1,\ldots,n^*\}$, 
% For all $i=1,\ldots,n^*$,
for all $n\ge n^*$
we have
\begin{equation}
\label{eq:dec_n_nhat}
\Pr( f_2(n)=r, f_3(n)=\kappa ) = q_{r,\kappa},
\end{equation}
where $f_j$ refers to the $j$-th component of $f$.
With this notation,
quantity $W_{f(n)}^{(k)}(\pi)$
is the waiting time of the $n$-th job arriving to the dispatcher induced by a policy $\pi\in \A_q(k)$, for all $n\ge n^*$.

% Let $W^{(k)}(\pi)$ be an additional random variable defined as follows: it equals $W_{r,\kappa}^{(k)}(\pi)$ with probability (independently of any other event)~$q_{r,k}$ for all~$r,\kappa$. Thus,
% $W^{(k)}(\pi)$ is the finite mixture of the $W_{r,\kappa}^{(k)}(\pi)$'s with weights~$q$.
% Next theorem says that $W^{(k)}(\pi)$ can be interpreted as the right random variable describing the stationary waiting time of jobs achieved with policy $\pi\in \A_q(k)$.
% It is proven by adapting the framework developed by Loynes in \cite{Loynes}.
% Convergence in distribution (in probability) is denoted by $\xrightarrow[]{d}$ (respectively, $\xrightarrow[]{\Pr}$).

Now, let $(Q_{r,\kappa})_{\forall r,\kappa}$ denote a partition of set $\{1,\ldots,n^*\}$.
The subsets $Q_{r,\kappa}$, for all $r$ and $k$, are thus disjoint and we further require that the number of points in $Q_{r,\kappa}$ is $n^*\,q_{r,\kappa}$.
Let
\begin{equation}
W^{(k)}(\pi)\bydef \sum_{r=1}^R \sum_{\kappa=1}^k
\mathbf{1}_{\{V\in Q_{r,\kappa}\}}
W_{r,\kappa}^{(k)}(\pi).
\end{equation}
Since $V$ is independent of the $W_{r,\kappa}^{(k)}(\pi)$'s,
the distribution of $W^{(k)}(\pi)$ is the finite mixture of the $W_{r,\kappa}^{(k)}(\pi)$'s with weights~$q$.
Next theorem says that $W^{(k)}(\pi)$ can be interpreted as the right random variable describing the stationary waiting time of jobs achieved with policy $\pi\in \A_q(k)$.
It is proven by adapting the framework developed by Loynes in \cite{Loynes}.
In the remainder of the paper, convergence in distribution (in probability) is denoted by $\xrightarrow[]{d}$ (respectively, $\xrightarrow[]{\Pr}$).

\begin{theorem}
\label{th0}
Let Case~\ref{as0} hold.
Let $q$ be such that $q_{r,\kappa} \lambda<\mu_r$ for all $r,\kappa$ and $\pi\in \A_q(k)$.
Then,
\begin{equation}
\label{eq:st_W}
W_{f(n)}^{(k)}(\pi)\le_{st} W_{f(n+1)}^{(k)}(\pi)
\end{equation}
and
\begin{equation}
\label{eq:wc_W}
W_{f(n)}^{(k)}(\pi) \xrightarrow[n\rightarrow\infty]{d} W^{(k)}(\pi).
\end{equation}
\end{theorem}

By using the monotone convergence theorem, Theorem~\ref{th0} implies that also the moments of $W_{f(n)}^{(k)}(\pi)$ converge to the moments of $W^{(k)}(\pi)$, provided that they are finite.

Finally, for a given $p\in \mathbb{R}_+^R$, we define the auxiliary random variable $W_r(p)$,
which corresponds to the stationary waiting time of a D/GI/1 queue with inter-arrival times
$(T_{n,r}(p))_{n\in\mathbb{N}}$ where $T_{n,r}(p)={\|p\|}/{(p_r\lambda)}$
and service times $(S_{n,1,r})_{n\in\mathbb{N}}$,  for all $r$.
% Within the foregoing assumptions on service times,
% it is known that the first two moments of $W_r(p)$ are finite if $\frac{\lambda}{\mu_r}\frac{p_r}{\|p\|}<1$ \cite[pg.~270]{Asmussen}. 
% Let also $W(p)$ be defined as $W_{r}(p)$ with probability (independently of any other event)~$p_r/\|p\|$ for all~$r$.
Let also $\overline{V}\bydef\overline{V}(p)$ be a random variable with values in $\{1,\ldots,R\}$ independent of any other random variable and such that $\Pr(\overline{V}=r)=p_r/\|p\|$, for all $r$, and
let
\begin{equation}
\label{eq:W_p_ascs}
W(p)\bydef \sum_{r=1}^R \mathbf{1}_{\{\overline{V}=r\}} W_{r}(p).
\end{equation}
Note that $\E W(p)=\sum_r \frac{p_r}{\|p\|} \E W_r(p)$
can be interpreted as the mean waiting time of $R$ independent D/GI/1 queues averaged over weights~$p/\|p\|$.

We will be interested in establishing convergence results for $W^{(k)}$ when $k\to\infty$.
With respect to some policies to be defined, we will show forms of convergence to $W(p)$.

\subsection{Discussion}
\label{eq:discussion}

% The model above has been studied in a number of works and we refer to the introduction for the related work.
Some remarks about the model above and Theorem~\ref{th0} follow:

\begin{itemize}

\item

To agree on a definition for $\E W^{(k)}$ or other moments of $W^{(k)}$, one does not necessarily need to know the distribution of $W^{(k)}$, e.g., one can achieve that using Cesaro sums \cite{Gau2000}. Matter of fact, existing works agree on the structure of $\E W^{(k)}$ \emph{without} constructing the distribution of $W^{(k)}$.
On the other hand, our approach needs to know the distribution of $W^{(k)}$ (and thus Theorem~\ref{th0}) because we can prove convergence results for $\E W^{(k)}$ \emph{only} through a distributional convergence argument, that is \cite[Theorem 3.5, pp. 31]{Bill}.
In particular, to prove $\E W^{(k)}$ converges to $\E W(p)$, we will prove convergence in distribution of $W^{(k)}$ to $W(p)$ and then the uniform integrability of the sequence of the $W^{(k)}$'s.
This is the reason why we need the characterization of the distribution of the stationary waiting time $W^{(k)}$, which we give in Theorem~\ref{th0} and prove through classical arguments.

\item
Though several works focused on finding policies that minimize the expected value $\E W^{(k)}$ (see the introduction), the analysis of $\E W^{(k)}$ in our scaling where $k\to\infty$ seems new.

\item
% We are interested in (deterministic) policies that are periodic, and
We assume that $q$ is a vector of rational numbers.
From a practical standpoint, this is not a loss of generality for obvious reasons. As an additional remark to support this assumption, we note that it has been proven that in several cases the $q$-vector that minimize $\min_{\pi\in\A_q(k)} \E W^{(k)}(\pi)$ is indeed rational \cite[Theorem~32, p.136]{Gau2000}; see also \cite{van2003structure}.

Our approach needs this assumption to prove \eqref{eq:wc_W}.
In the case where $\pi(n)$ is not periodic, a case that we do not consider here, and
the limit 
% \begin{equation}
$\lim_{n\to\infty} \frac{1}{n} \sum_{n'=1}^n \mathbf{1}_{\{\pi(n')=(r,\kappa)\}}$
% \end{equation}
does not exist,
it would be interesting to know whether some convergence in distribution of $W_{f(n)}^{(k)}(\pi)$ occurs.
A totally different argument will be needed here.
% We do not investigate this issue in this paper.

% \item
% One may regard a policy as a function that maps each natural number $n$ to a probability mass function over the set of queues (this may depend on~$n$). 
% In our model, it is assumed that such probability functions put the whole mass on exactly one queue.
% Such assumption is only made for simplicity and lets us simplify notation.
% We claim that the results in this paper holds even when such probability functions put mass on multiple queues. 
% We will briefly discuss this claim in Section~\ref{eq:main_res}, though we do not present it rigorously for simplicity.

\item
The fact that in Case~\ref{as0} we require that {Var}$\, T_{n}^{(k)} = o(k)$ is not a loss of generality for Theorem~\ref{th0}, as we do not let $k\to\infty$ there.
Case \ref{as0} covers the case where $T_{n}^{(k)}$ has a $(Gk,\alpha)$-phase-type distribution
% with mean $(\lambda k)^{-1}$,
where both $G$ and $\alpha$ are not functions of~$k$.
% Since phase-type distributions are dense in the field of positive-valued distributions, Case \ref{as0} covers a large number of practical scenarios.

\item
We may have assumed that each queue of type $r$ was replicated $k z_r + o(k)$ times instead of just $k$ times.
This is essentially equivalent to our setting and our proofs can be easily adapted to this case though we do not do it for simplicity.

\end{itemize}

\section{Main results}
\label{sec:mainres}

For $p\in\mathbb{Z}_+^R$, let
\begin{equation}
\mathcal{P}_p(k)\bydef \left\{q\in\mathbb{Q}_+^R\times \mathbb{Q}_+^k:  \sum\limits_{\kappa=1}^{k} \tfrac{q_{r,\kappa}}{k}= \tfrac{p_r}{\|p\|},\, \forall r\right\}
\end{equation}
and
\begin{equation}
\overline{\A}_p(k)\bydef \bigcup_{ q\in \mathcal{P}_p(k)} \A_q(k).
\end{equation}
The set $\overline{\A}_p(k)$ is interpreted as the set of all periodic policies for which $\lambda k \, \frac{p_r}{\|p\|}$ is the aggregate mean arrival rate of jobs to all type-$r$ queues.
We consider the following problem.
\begin{problem}
\label{prob1}
Let $p\in\mathbb{Z}_+^R$ be given.
Determine the optimizers and the optimal objective function value of 
\begin{equation}
\min_{\pi\in\overline{\A}_p(k)} \E W^{(k)}(\pi).
\end{equation}
\end{problem}
As discussed in the introduction, this is considered as a difficult problem.
We are interested in establishing structural properties of Problem~\ref{prob1} when $k$ is large.
In the following, we define a certain class of policies and then present our main results.

% Within this set, we present our main result, which says that the average stationary waiting time $W^{(k)}$
% converges to the mean waiting time of~$R$ independent D/GI/1 queues, \emph{independently of the sequence of policies chosen in $\Cup$.
% This implies that all the policies in our set are \emph{asymptotically equivalent} or \emph{asymptotically optimal}, as they yield the same waiting time.

\subsection{A class of periodic policies}
\label{sec:per_pol}

For $p\in\mathbb{Z}_+^{R}$, we define $\J_p^{(k)}$ as the subset of all policies
$\pi\in \overline{\A}_p(k)$
% $\pi\in \bigcup_{\substack{q\in\Z_+^R\times \Z_+^k: \sum_\kappa q_{r,\kappa}=k p_r\, \forall r}} \A(k)$
that satisfy
the following properties:
\begin{itemize}
% IMPLIED BY THE SECOND AND THE THIRD
% \item Among every $k \|p\|$ consecutively-arriving jobs, exactly $p_r$ jobs are sent to each queue of type~$r$, for each $r$. That is,
% \begin{equation}
% \nonumber
% \sum_{n'=n+1}^{k\|p\|+n} \mathbf{1}_{\{\pi(n')=(r,\kappa)\}}=p_r,\quad \forall n.
% \end{equation}

% REWRITTEN BELOW
% \item Among every $\|p\|$ consecutively-arriving jobs, exactly $p_r$ jobs of type~$r$ arrive. That is,
% \begin{equation}
% \nonumber
% \sum_{\kappa=1}^k \sum_{n'=n+1}^{\|p\|+n} \mathbf{1}_{\{\pi(n')=(r,\kappa)\}}=p_r,\quad \forall n.
% \end{equation}
\item
The sequence $(\pi_1(n))_{n\in\mathbb{N}}$ has period $\|p\|$ and $p_r$ is the number of jobs sent to type-$r$ queues per period, for $1 \le r \le R$.

\item Let $n_{1},n_{2},\ldots$ be the subsequence of all jobs that are sent to queues of type $r$.
Then, $\pi_2(n_1)=1$, and $\pi(n_{j+1}) = \pi(n_{j}) + (0,1)$ if $\pi_2(n_{j}) < k$
and $\pi(n_{j+1}) = (r,1)$ otherwise.

% \item The first job that visits a queue of type $r$ is sent to queue 1, i.e., $\pi_2(\min\{n:\pi_1(n)=r \})=1$.
% \item If $n$ is the first job such that $\pi_1(n)=r$ then $\pi_2(n)=1$ (equivalently, $\pi_2(\min\{n:\pi_1(n)=r \})=1$), for all $r$.
\end{itemize}

% The first property says that $\pi$ is periodic with period $k\|p\|$ and that in a period exactly $p_r$ jobs are sent to each queue of type $r$.
% The second property says that arrivals remain ``well distributed'' among the different types when~$k$ increases.
% Finally, the third property says that queues of type~$r$ are accessed by jobs in a round-robin (or cyclic) order starting from queue 1, and we need it to ensure that the cardinality of $\J_p^{(k)}$, i.e., $|\J_p^{(k)}|$, does not
% vary with $k$.

The first property specifies the periodicity of $\pi$ with respect the types of queues, while the second with respect to queues.
Thus, the second property says that queues of type~$r$ are accessed by jobs in a round-robin (or cyclic) order starting from queue 1, and we need it to ensure that the cardinality of $\J_p^{(k)}$, i.e., $|\J_p^{(k)}|$, does not
vary with $k$.
One can verify that
\begin{equation}
 |\J_p^{(k)}|=\frac{\|p\|!}{\prod_r (p_r!)}.
\end{equation}

These policies can be implemented in a distributed manner.
More precisely, one can think that there are two tiers of dispatchers:
the dispatcher in the first tier schedules jobs \emph{inter}-group,
while the dispatchers in the second tier, $R$ in total, schedule jobs \emph{intra}-group and implement round-robin.

With respect to a sequence of policies $(\pi^{(k)})_{k\in\mathbb{N}}$, where $\pi^{(k)}\in \J_p^{(k)}$,
% it can be shown that (see \eqref{eq:lim_a_jkr2})
we will show (Lemma~\ref{lm:TSL05})
\begin{equation}
T_{n,\kappa,r}^{(k)}(\pi^{(k)}) \xrightarrow[k\rightarrow\infty]{\Pr} \tfrac{\|p\|}{p_r\lambda},\qquad \forall n.
\end{equation}
% We note that $T_{n,\kappa,r}^{(k)}(\pi^{(k)})$ varies with $n$ and that the former limit will follow by~\eqref{eq:lim_a_jkr}. 
This means that the finite dimensional distributions of the arrival process at each queue will be `close' to the deterministic process, when $k$ is large, which implies that some form of convergence to $W(p)$ should occur in view of the continuity of the stationary waiting time~\cite{borovkov1976stochastic}.
% This intuition can be easily seen by considering one type of queues ($R=1$) and the round-robin policy,
% which in this specific case is known to minimize $\E W^{(k)}$ \cite{LR98}.
% Here, interarrival times are Erlang distributed with $k$ phases each having rate $k\lambda$, which means that their densities concentrate more and more around $1/\lambda$ as $k$ grows.

\subsection{Asymptotic equivalence and optimality}
\label{eq:main_res}

Next theorem proves a first form of convergence of $W^{(k)}(\pi^{(k)})$ to $W(p)$.

\begin{theorem}
\label{th05}
Let Case~\ref{as0} hold.
Let $p\in\mathbb{Z}_+^{R}$ be such that $\lambda \frac{p_r}{\|p\|}<\mu_r $ for all~$r$.
Let also an arbitrary sequence $(\pi^{(k)})_{k\in\mathbb{N}}$ be given where $\pi^{(k)}\in\J_p^{(k)}$ for all $k$.
Then
\begin{equation}
\label{eq:lim_Dist_Wk_pqsa}
W^{(k)}(\pi^{(k)}) \xrightarrow[k\rightarrow\infty]{d}  W(p).
\end{equation}
\end{theorem}

% Theorem~\ref{th05} proves the convergence in distribution of the random variable $W^{(k)}(\pi^{(k)})$.
Theorem~\ref{th05} is not enough to claim that $\E W^{(k)}(\pi^{(k)})$ converges as well.
Convergence of the first moment is important from an operational standpoint, as in practice one desires to optimize over $\E W^{(k)}$ or $\Var\, W^{(k)}$. 
Under some additional assumptions,
next theorem states that also the expected value and the variance of $W^{(k)}$ converge.
Furthermore, it states that all the policies in set $\J_p^{(k)}$ are both {asymptotically equivalent} and {asymptotically optimal}, with respect to the criterion in Problem~\ref{prob1}.

\begin{theorem}
\label{th1}
Let Case~\ref{as1} or \ref{as2} hold.
Let $p\in\mathbb{Z}_+^{R}$ be such that $\lambda \frac{p_r}{\|p\|}<\mu_r $ for all~$r$.
Let also an arbitrary sequence $(\pi^{(k)})_{k\in\mathbb{N}}$ be given where $\pi^{(k)}\in\J_p^{(k)}$ for all $k$.
If $\E [(S_{1,\kappa,r}^{(k)})^3]<\infty$, then
\begin{equation}
\label{eq:MLB}
% \liminf_{k\to\infty} \min_{\substack{q\in\mathbb{R}_+^R\times \mathbb{R}_+^k: \\\sum_\kappa \frac{q_{r,\kappa}}{k}= \frac{p_r}{\|p\|}\, \forall r}} \min_{\pi\in\A_q(k)} \E W^{(k)}(\pi) =
% \lim_{k\to\infty} \E W^{(k)}(\pi^{(k)}).
% 
\liminf_{k\to\infty} \min_{\pi\in \overline{\A}_p(k)} \E W^{(k)}(\pi) =
\lim_{k\to\infty} \E W^{(k)}(\pi^{(k)})
\end{equation}
\begin{equation}
\label{eq:lim_E_Wk_pqsa}
\lim_{k\to\infty} \E W^{(k)}(\pi^{(k)}) = \E W(p).
\end{equation}
Furthermore, if $\E [(S_{1,\kappa,r}^{(k)})^5]<\infty$, then
\begin{equation}
\label{eq:lim_Var_Wk_pqsa}
\lim_{k\to\infty} \mbox{\emph{Var}} \, W^{(k)}(\pi^{(k)}) = 
\sum_{r} \tfrac{p_r}{\|p\|}   \left( \mbox{\emph{Var}} \, W_{r}(p) +   \left( \E W_{r}(p) - \E W(p) \right)^2 \right).
\end{equation}
\end{theorem}

Provided that~$k$ is large, thus, no other policy in $\overline{\A}_p(k) \setminus \J_p^{(k)}$ can do better than any $\pi^{(k)}\in\J_p^{(k)}$ to minimize $\E W^{(k)}$.
It also explicits the limiting value of $\E W^{(k)}(\pi^{(k)})$, which is $\E W(p)$.
It is known that $\E W(p)$ is a convex function in $p$ (e.g., \cite{NeelyM05}).

% The assumption that $\lambda \frac{p_r}{\|p\|}<\mu_r $ for all $r$ is needed to make the system ergodic.

% To prove Theorem~\ref{th1}, we use the lower bound for the G/G/1 queue developed in \cite{humblet1982determinism}, which is expressed in terms of a D/G/1 queue.
% Using this lower bound first and then that the waiting time of the D/GI/1 queue is convex increasing in its arrival rate, it is not difficult to show that
% \begin{equation}
% \label{eq:MLB2}
% \E W (p) \le \E W^{(k)}(\pi)
% \end{equation}
% for any $\pi\in\overline{\A}_p(k)$.
% Then, we prove that the sequence $\E W^{(k)}(\pi^{(k)})$ is upper bounded by a sequence that converges to the lower bound in \eqref{eq:MLB2}.

To prove Theorem~\ref{th1}, we use the well-known fact that \cite{Asmussen}
\begin{equation}
\label{eq:icx_W_r}
W_{r}(p) \le_{icx} W_{r,\kappa}^{(k)}(\pi)
\end{equation}
for any $\pi\in\overline{\A}_p(k)$, where $\le_{icx}$ denotes the increasing-convex order (see, e.g., \cite{shantikumar,stoyan1983comparison} for their definition).
Using this lower bound first and then that the waiting time of the D/GI/1 queue is convex increasing in its arrival rate, it is not difficult to show that
\begin{equation}
\label{eq:MLB2}
\E W (p) \le \E W^{(k)}(\pi).
\end{equation}
Then, we prove that the sequence $\E W^{(k)}(\pi^{(k)})$ is upper bounded by a sequence that converges to the lower bound in \eqref{eq:MLB2}.
An observation here is that the lower bound \eqref{eq:icx_W_r} holds under conditions that are weaker than those assumed in this paper; see \cite{humblet1982determinism}.
For instance, it is possible to extend~\eqref{eq:MLB2} (and thus Theorem~\ref{th1})
to the case where
i) policies are not periodic,
ii) the fractions of jobs to send in each queue are not necessarily rational numbers, and
iii) policies are randomized
% \footnote{In this case, an expectation appears in~\eqref{eq:alpha}.}
\cite{Puterman1994},
that is the case where $\pi(n)$ is any probability mass function over the set of queues. 
We do not investigate these extensions in further detail.

\section{Proofs}
\label{sec:proofs}

In this section, we develop proofs for Theorems~\ref{th0},~\ref{th05} and~\ref{th1}.
Before doing this, we fix some additional notation and show how it is possible to make the arrival process at each queue stationary with respect to any policy in $\pi^{(k)}\in \mathcal{A}_q(k)$, $q\in \mathbb{Q}_+^{kR}$ such that $\|q\|=1$  (as assumed in Section~\ref{sec:model}).

Let us consider the $\kappa$-th queue of type $r$ and its arrival process $(T_{n,\kappa,r}^{(k)})_{n\in\mathbb{N}}$. Each inter-arrival time clearly depends on the policy $\pi^{(k)}$ implemented by the dispatcher, i.e., $T_{n,\kappa,r}^{(k)} = T_{n,\kappa,r}^{(k)}(\pi^{(k)})$, though in the following we drop such dependence for notational simplicity.
Since $\pi^{(k)}$ is periodic by construction with period $n^*$ and $n^* q_{r,\kappa}$ jobs have to be sent within a cycle to the queue identified by the couple $(r,\kappa)$,
the sequence $(T_{n,\kappa,r}^{(k)})_{n\in\mathbb{N}}$ is composed of a repeated pattern
of $n^* q_{r,\kappa}$ inter-arrival times, that we can write as
\begin{equation}
\label{eq:pattern_general}
%  A_{1,\kappa,r}^{(k)},A_{2,\kappa,r}^{(k)}, \ldots, A_{p_r,\kappa,r}^{(k)},
 A_{1,\kappa,r}^{(k)},A_{2,\kappa,r}^{(k)}, \ldots, A_{n^* q_{r,\kappa},\kappa,r}^{(k)},
\end{equation}
where each quantity $A_{j,\kappa,r}^{(k)}$, $j=1,\ldots,n^* q_{r,\kappa}$,
is the sum of a deterministic number (that depends on $\pi^{(k)}$) of inter-arrival times to the dispatcher.
We denote such number by
$a_{j,r,\kappa}^{(k)}$.
Thus,
\begin{equation}
\label{eq:dist_Aj}
A_{j,\kappa,r}^{(k)} =_{st} \sum_{n=1}^{a_{j,r,\kappa}^{(k)}} T_{n}^{(k)},
\end{equation}
where $=_{st}$ denotes equality in distribution.

If $\pi^{(k)}\in \J_p^{(k)}$ for some $p\in \mathbb{Z}_+^R$, then we notice that $a_{j,r,\kappa}^{(k)}$ does not vary with $\kappa$ because by symmetry the arrival processes of all queues of a given type are equal, in distribution, up to a shift in time.
In this case, the arrival process at \emph{any} queue of type $r$ becomes a sequence composed of a repeated pattern
of $p_r$ inter-arrival times that we can write as
\begin{equation}
\label{eq:pattern}
A_{1,\kappa,r}^{(k)},A_{2,\kappa,r}^{(k)}, \ldots, A_{p_r,\kappa,r}^{(k)}.
\end{equation}
Therefore, when $\pi^{(k)}\in \J_p^{(k)}$, we will just write $a_{j,r}^{(k)}$ instead of $a_{j,r,\kappa}^{(k)}$ and it is also clear that
\begin{equation}
\label{eq:prop_a_jrk}
\sum_{j=1}^{p_r} a_{j,r}^{(k)}=k\|p\|
\end{equation}
and that 
% using \eqref{eq:prop_a_jrk} we obtain
\begin{equation}
\label{eq:exp_Tn}
\E T_{n,\kappa,r}^{(k)}
% = \tfrac{1}{p_r} \sum_{j=1}^{p_r} \E A_{j,\kappa,r}^{(k)}
=\tfrac{1}{p_r}\sum_{j=1}^{p_r} a_{j,r}^{(k)} \, \E T_{n}^{(k)}
=\tfrac{1}{p_r}\sum_{j=1}^{p_r} \tfrac{a_{j,r}^{(k)}}{k\lambda}
=\tfrac{\|p\|}{p_r\lambda}.
\end{equation}

Now, we want to make $(T_{n,\kappa,r}^{(k)})_{n\in\mathbb{N}}$ stationary.
This can be done as follows by randomizing over the first inter-arrival time of queue $(r,\kappa)$.
Now, let us consider the auxiliary random variables $U_{r,\kappa}$, for all $r$ and $\kappa$, which we assume independent each other and of any other random variable and having a uniform distribution in $[0,1]$.
Then, we take
% \begin{equation}
% \label{eq:def_T1}
% T_{1,r,\kappa}^{(k)} \bydef 
% \left\{
% \begin{array}{ll}
% A_{1,r,\kappa}^{(k)} ,  & \mbox{if }  U_{r,\kappa} \in \left[0,\frac{1}{p_r}\right) \\
% A_{2,r,\kappa}^{(k)} ,  & \mbox{if }  U_{r,\kappa} \in \left[\frac{1}{p_r},\frac{2}{p_r}\right)\\
% \ldots &\\
% A_{p_r,r,\kappa}^{(k)}, & \mbox{if }  U_{r,\kappa} \in \left[\frac{p_r-1}{p_r},1\right].
% \end{array}
% \right.
% \end{equation}
\begin{equation}
\label{eq:def_T1}
% particular case
% T_{1,r,\kappa}^{(k)} \bydef \sum_{j=1}^{p_r} A_{j,r,\kappa}^{(k)} \mathbf{1}_{\left\{ U_{r,\kappa}\in \left[\frac{j-1}{p_r},\frac{j}{p_r}\right]\right\}}
% General case
T_{1,r,\kappa}^{(k)} \bydef \sum_{j=1}^{n^* q_{r,\kappa}} A_{j,r,\kappa}^{(k)} \mathbf{1}_{\left\{ U_{r,\kappa}\in \left[\frac{j-1}{n^* q_{r,\kappa}},\frac{j}{n^* q_{r,\kappa}}\right]\right\}}
\end{equation}
Therefore, if $T_{1,r,\kappa}^{(k)}=A_{j,r,\kappa}^{(k)}$, for some $j<n^* q_{r,\kappa}$, then $T_{2,r,\kappa}^{(k)}=A_{j+1,r,\kappa}^{(k)}$ and so forth according to the pattern~\eqref{eq:pattern_general}.
Defined in this manner, one can see that $(T_{n,r,\kappa}^{(k)})_{n\in\mathbb{N}}$ is stationary as desired.
At this point, the issue is the following.
Consider two queues, say $(r,\kappa)$ and $(r',\kappa')$,
and suppose that
$T_{1,r,\kappa}^{(k)} = A_{j,r,\kappa}^{(k)}$
and 
$T_{1,r',\kappa'}^{(k)} = A_{j',r',\kappa'}^{(k)}$.
We should check whether policy $\pi^{(k)}$ is actually able to induce arrival processes at the queues
% , i.e., $(T_{n,r,\kappa}^{(k)})_{n\in\mathbb{N}}$ for all $r$ and $\kappa$,
equal (samplepath-wise) to the ones built above through \eqref{eq:def_T1}.
One can easily see that this can be done with a possible shift of time for the arrival process at the queues and possibly discarding a finite number of jobs. This is allowed because these operations do not change the stationary behavior.

% and that this also holds for a generic queue $(r,\kappa)$.
% by properly shifting in time its first arrival.

In the remainder, we will use stochastic orderings. 
We will denote by $\le_{st}$, $\le_{cx}$ and $\le_{icx}$, the \emph{usual stochastic order}, the \emph{convex order} and the \emph{increasing convex order}, respectively; we point to, e.g., \cite{shantikumar,stoyan1983comparison} for their definition.

We will also refer to the following lemma, which can be easily proven.

\begin{lemma}
\label{lm:ddard}
Let $N$ be a finite positive integer and suppose $(f_{k,n} )_{(k,n)\in \mathbb{N}\times \{1,...,N\}}$ is a semi-infinite array of numbers such that for some constant $c$, $\lim_{k\to\infty} f_{k,n} = c$, for $n\in \{1, \ldots, N \}$. Then,
for any sequence $(n_k )_{k\in\mathbb{N}}$ with values in $\{1, \ldots, N \}$, $\lim_{k\to\infty} f_{k,n_k} = c$.
\end{lemma}

% \begin{lemma}
% \label{lm:ddard}
% Let $(D_k)_{k\in\mathbb{N}}$ be a sequence of ordered finite sets having $N>1$ elements each.
% For $n\in D_k$, let $\#n\in\mathbb{N}$ denote the position of $n$ in $D_k$ (e.g., if $D_k=\{c,a,b\}$ then $\#a=2$).
% Let $(n_k)_{k\in\mathbb{N}}$ be a sequence such that $n_k\in D_k$, for all $k$.
% Let $(f_k)_{k\in\mathbb{N}}$ be a sequence of functions where $f_k:D_k\to\mathbb{R}$.
% % 
% If $\lim_{k\to\infty} f_k(n_k)=c$ for all sequences $(n_k)_{k\in\mathbb{N}}$ such that $\#n_k=\#n_1$, for all $n_1\in D_1$,
% then $\lim_{k\to\infty} f_k(n_k)=c$ for all sequences $(n_k)_{k\in\mathbb{N}}$.
% \end{lemma}

% In words,
% previous lemma says that if we study the limit when $k\to\infty$
% of the first value of function $f_k$,
% then of the second value of function $f_k$,
% $\ldots,$
% then of the $N$-th value of function $f_k$,
% and it happens that these limits are all equal to $c$, then 
% the limit of $f_k(n_k)$ must be $c$ for any choice of the $n_k$.

We now give proofs for our results, i.e., Theorems~\ref{th0},~\ref{th05} and~\ref{th1}.

\subsection{Proof of Theorem~\ref{th0}}

% Given policy $\pi$ and queue $(r,\kappa)$, we say that $(r',\kappa')$ {precedes} $(r,\kappa)$ at time $n$
% if $\pi(n)=(r,\kappa)$ and $\pi(n-1)=(r',\kappa')$.

Let random variables $V_1$ and $V_2$ be given such that
$V_1 =_{st} V$
and
$V_2=V_1-1$ if $V_1>1$ otherwise $V_2=n^*$.
We prove \eqref{eq:st_W} through Strassen's theorem building a coupling $(\tilde{W}_{f(n)}^{(k)}(\pi),\tilde{W}_{f(n+1)}^{(k)}(\pi))$ of $W_{f(n)}^{(k)}(\pi)$ and $W_{f(n+1)}^{(k)}(\pi)$ through $V_1$ and $V_2$ ensuring that $\tilde{W}_{f(n)}^{(k)}(\pi)\le \tilde{W}_{f(n+1)}^{(k)}(\pi)$.
This is done as follows. First, let $\tilde{W}_{f(n)}^{(k)}(\pi)$ and $\tilde{W}_{f(n+1)}^{(k)}(\pi)$ be the Loynes waiting times (see \cite[p.~501]{Loynes}) obtained when the first queue to serve a job is given by the outcome of $V_1$ and $V_2$, respectively;
thus, $\tilde{W}_{f(n)}^{(k)}(\pi)$ is the waiting time at time $0$ with $n'$ jobs in the past at queue $(r,\kappa)$, provided that $f(n)=(n',r,\kappa)$.
Then, we let the (Loynes) waiting times be driven by the same realizations of the random inter-arrival and service times.

% We now prove \eqref{eq:wc_W}.
% Since $\pi$ is periodic with period $n^*$,
% we first observe that 
% $\pi$ returns the same queue along subsequence $(n \,n^* + i)_{n\in \mathbb{N}}$, for all $i=1,\ldots,n^*$.
% However, such queue is not a priori known because it depends on the outcome of $V$.
% By construction, we have $\Pr( \pi(n \,n^* + i) = (r,\kappa)) = q_{r,\kappa}$. 
% Thus, conditioning on $\pi(n \,n^* + i)$ and using that $W_{n,r,\kappa}^{(k)}(\pi)$ converges in distribution to $W_{r,\kappa}^{(k)}(\pi)$ (\cite{Loynes}), we get
% \begin{subequations}
% \begin{align}
% % \lim_{n\to\infty} \Pr(W_{f(n)}^{(k)}(\pi) \le t)
%  \lim_{n\to\infty} \Pr(W_{f(n n^* +i)}^{(k)}(\pi) \le t)
% % 
% & = \lim_{n\to\infty} \sum_{r,\kappa}  q_{r,\kappa}   \Pr(W_{f(n n^* +i)}^{(k)}(\pi) \le t | \pi(n n^* +i)=(r,\kappa))\\
% % & = \lim_{n\to\infty} \sum_{r,\kappa}  q_{r,\kappa}   \Pr(W_{n,r,\kappa}^{(k)}(\pi) \le t | \pi(V)=(r,\kappa))\\
% % 
% & = \sum_{r,\kappa} q_{r,\kappa} \lim_{n\to\infty} \Pr(W_{n,r,\kappa}^{(k)}(\pi) \le t)\\
% % 
% & = \sum_{r,\kappa} q_{r,\kappa} \Pr(W_{r,\kappa}^{(k)}(\pi) \le t)\\
% % 
% & = \Pr(W^{(k)}(\pi) \le t).
% \end{align}
% \end{subequations}
% 
% Similarly, the second and third component of $f(n \,n^* + i)$ are equal to $\pi(n \,n^* + V)$.

We now prove \eqref{eq:wc_W}.
Since $\pi$ is periodic with period $n^*$,
we first observe that 
$\pi$ returns the same queue along subsequence $(n \,n^* + i)_{n\in \mathbb{N}}$, for all $i=1,\ldots,n^*$.
Similarly, also the second and third component of $f(n \,n^* + i)$ do not change along these subsequences, though they are not known in advance because they depend on the outcome of random variable $V$, see \eqref{eq:dec_n_nhat}.
Thus, for all $i=n^*+1,\ldots,2n^*$, by construction we have
\begin{equation}
\Pr( f_2(n \,n^* + i)=r, f_3(n \,n^* + i)=\kappa)
=\Pr( f_2(i)=r, f_3(i)=\kappa )
% =\Pr( f(i)=(1,r,\kappa))
=q_{r,\kappa},
\end{equation}
and we get
\begin{subequations}
\begin{align}
&\lim_{n\to\infty} \Pr(W_{f(n n^* +i)}^{(k)}(\pi) \le t)\\
% & = \lim_{n\to\infty} \sum_{r,\kappa}  q_{r,\kappa}   \Pr(W_{f(n n^* +i)}^{(k)}(\pi) \le t \,|\, f_2(n \,n^* + i)=r, f_3(n \,n^* + i)=\kappa)\\
\label{eq:th0_1}
& = \lim_{n\to\infty} \sum_{r,\kappa}  q_{r,\kappa}   \Pr(W_{n,r,\kappa}^{(k)}(\pi) \le t \,|\, (f_2(i),f_3(i))=(r,\kappa))\\
% \label{eq:th0_2}
% & = \sum_{r,\kappa} q_{r,\kappa} \lim_{n\to\infty} \Pr(W_{n,r,\kappa}^{(k)}(\pi) \le t)\\
\label{eq:th0_3}
& = \sum_{r,\kappa} q_{r,\kappa} \Pr(W_{r,\kappa}^{(k)}(\pi) \le t)\\
\label{eq:th0_4}
& = \Pr(W^{(k)}(\pi) \le t).
\end{align}
\end{subequations}
In \eqref{eq:th0_1}, we have conditioned on $f(i)$.
In \eqref{eq:th0_3}, we have used that $W_{n,r,\kappa}^{(k)}(\pi)$ converges in distribution to $W_{r,\kappa}^{(k)}(\pi)$; see \cite{Loynes}.
In \eqref{eq:th0_4}, we have used the definition of $W^{(k)}(\pi)$.
Now, since the limit in \eqref{eq:th0_4} does not depend on~$i$, the proof is concluded by applying Lemma~\ref{lm:ddard} once noted that $n^*$ is a finite positive integer.

\subsection{Proof of Theorem~\ref{th05}}

We first observe that we can prove this theorem under some assumption on the sequence $(\pi^{(k)})_{k\in\mathbb{N}}$.
Given $\pi^{(1)}\in\J_p^{(1)}$, we require that for all~$k$:
\begin{equation}
\label{eq:subseq_pik}
\pi_1^{(k)}(n) = \pi_1^{(1)}(n), \qquad \forall n.
\end{equation}
One may refer to these sequences as the `natural' scaling of policy $\pi^{(1)}$: in the two-tier interpretation of our policies, \eqref{eq:subseq_pik} means that the dispatcher at the first tier implements the same policy, to queue types, when $k$ grows.

These sequences will be assumed along this proof.
If this theorem holds for these sequences, then it also holds for all the sequences in view of Lemma~\ref{lm:ddard} and of the fact that the cardinality of $\J_p^{(k)}$ does not vary with~$k$.

For $m\in\mathbb{N}$, let
\begin{equation}
\label{eq:k_m}
k_m\bydef m\, \mbox{lcm}(p), 
\end{equation}
where lcm$(p)$ denotes the least common multiple of $p_1,\ldots,p_R$. 
The subsequences $(k_m+i)_{m\in\mathbb{N}}$, for all $i=1,\ldots,\mbox{lcm}(p)$, play a key role in our proof of Theorem~\ref{th1}.
Along these subsequences, next fact holds true and follows by construction of the policies in set $\J_p^{(k)}$: it is a direct consequence of the fact that queues of the same type are visited in a round-robin manner.
\begin{fact}
\label{f1}
% For $m>1$, $j=1,\ldots,p_r$ and $i=1,\ldots,\mbox{\emph{lcm}}(p)$, 
For $m>1$, $j=1,\ldots,p_r$ and $i\in\mathbb{N}$, 
% \begin{equation}
% 
$a_{j,r}^{(p_rm+i)} =a_{j,r}^{( p_r(m-1)+i)} + \|p\|$.
% \end{equation}
\end{fact}
\begin{proof}
By construction, we have $\pi_2^{(p_rm+i)}(V)=1$ and  $\pi^{(p_rm+i)}(V)=\pi^{(p_r(m-1)+i)}(V)$, which in some sense couples the arrival processes at queues of the $(p_rm+i)$-th and $(p_r(m-1)+i)$-th systems. 

Without loss of generality, let us assume that $(r,1)$ is the queue that receives the first job.

Now, Fact \ref{f1} holds true because $(p_rm+i) - (p_r(m-1)+i) = p_r$ jobs must be sent to some queues of type $r$ of the $(p_rm+i)$-th system in the time interval $[V+ a_{j,r}^{( p_r(m-1)+i)}, V+a_{j,r}^{( p_r m+i)}-1]$, for $j=1$, and the number of arrivals at the dispatcher in that interval is exactly $\|p\|$ by construction of the policies in $\J_p^{(k)}$ for any $k$ (see subsection~\ref{sec:per_pol})

This argument applies to all the other queues because we have considered periodic policies.
\end{proof}

We show Fact \ref{f1} in the following example, to help understanding its meaning and proof.

\begin{example}
Assume $R=2$, $p=(3,2)$, $r=1$, $i=1$, $m=2$, $j=1$.
Assume also that $\pi\in\J_p^{(1)}$ is such that $(\pi_1(V+n))_{n=0,\ldots,4}=(1,1,2,1,2)$.
Then, the sequence of queues to be visited for both systems $p_rm+i$ and $p_r(m-1)+i$ is given in Table \ref{table1}, where we can see that the decomposition in Fact~\ref{f1} holds.

\begin{table}[h]
\label{table1}
\begin{tabular}{c@{ }c@{ }c@{ }c@{ }c@{ }c@{ }c@{ }c@{ }c@{ }c@{ }c@{ }c@{ }c@{ }}
\multicolumn{1}{c}{} & \multicolumn{11}{c}{$ a_{j,r}^{(p_rm+i)}$}\\
\multicolumn{1}{c}{} & \multicolumn{11}{c}{\downbracefill}\\[-1ex]
& & & & & & & & & & &  \\
$p_rm+i$: & (1,1) & (1,2) & (2,1) & (1,3) & (2,2) & (1,4) & (1,5) & (2,3) & (1,6) & (2,4) & (1,7) & (1,1) \ldots  \\
\multicolumn{7}{c}{} & \multicolumn{5}{c}{\upbracefill}\\[-1ex]
\multicolumn{7}{c}{} & \multicolumn{5}{c}{$ \|p\| $}\\
& & & & & & & & & & &  \\
$p_r(m-1)+i$: & (1,1) & (1,2) & (2,1) & (1,3) & (2,2) & (1,4) & (1,1) & (2,3) & (1,2) & (2,4) & (1,3) & (1,4) \ldots  \\
\multicolumn{1}{c}{} & \multicolumn{6}{c}{\upbracefill}\\[-1ex]
  \multicolumn{1}{c}{} & \multicolumn{6}{c}{$ a_{j,r}^{( p_r(m-1)+i)}$}\\
\end{tabular}
\caption{Illustrative example for the decomposition in Fact~\ref{f1}.}
\end{table}

\end{example}

Since Fact~\ref{f1} holds for any $m>1$, we can make the replacement $m\to m \tfrac{\mbox{lcm}(p)}{p_r}$ for which we obtain
\begin{equation}
% a_{j,r}^{(p_rm+i)} =a_{j,r}^{( p_r(m-1)+i)} + \|p\|
% a_{j,r}^{(p_rm \tfrac{\mbox{lcm}(p)}{p_r}+i)} =a_{j,r}^{( p_r(m \tfrac{\mbox{lcm}(p)}{p_r}-1)+i)} + \|p\|
% a_{j,r}^{(m \mbox{lcm}(p)+i)} =a_{j,r}^{( m \mbox{lcm}(p) - p_r + i)} + \|p\|
a_{j,r}^{(k_m +i)} =a_{j,r}^{( k_m - p_r + i)} + \|p\|.
\end{equation}
Unfolding this recursion, for $m\in\mathbb{N}$, $i=1,\ldots,\mbox{{lcm}}(p)$, we get
\begin{equation}
\label{eq:sim_ordered0}
a_{j,r}^{(k_{m}+i)} 
=a_{j,r}^{( k_m - 2p_r + i)} + 2\|p\|
=\cdots
= a_{j,r}^{(k_{m-1}+i)} +  \tfrac{\mbox{{lcm}}(p)}{p_r}\|p\|  = a_{j,r}^{(i)} + m  \tfrac{\|p\|}{p_r}\mbox{{lcm}}(p)
\end{equation}
and therefore
\begin{equation}
\label{eq:sim_ordered}
\frac{a_{j,r}^{(k_m+i)}}{k_m+i}
= \frac{a_{j,r}^{(i)} + m  \tfrac{\|p\|}{p_r}\mbox{{lcm}}(p)}{ m\cdot \mbox{lcm}(p)+i}\xrightarrow[m\rightarrow\infty]{} \frac{\|p\|}{p_r}.
\end{equation}

As a technical observation, in \eqref{eq:sim_ordered0} we note why we require index $i$ to range in $\{1,\ldots,\mbox{{lcm}}(p)\}$: it `closes' the recursion.

Since \eqref{eq:sim_ordered} holds for all $i=1,\ldots,\mbox{lcm}(p)$,  by using Lemma~\ref{lm:ddard} we obtain
\begin{equation}
\label{eq:lim_a_jkr2}
\lim\limits_{k\rightarrow \infty} \tfrac{a_{j,r}^{(k)}}{k}=\tfrac{\|p\|}{p_r}.
\end{equation}

As a comment, we note here that \eqref{eq:lim_a_jkr2} could be proven without Fact~\ref{f1} and what has followed.
% However, we stress that Fact~\ref{f1} will play a crucial role later in the proof of Theorem~\ref{th1} (see Lemma~\ref{lm:TSL}.iii).
However, we stress that we will need Fact~\ref{f1} later anyhow, as it will play a crucial role in the proof of our main result Theorem~\ref{th1} (see Lemma~\ref{lm:TSL}.iii).

\begin{lemma}
\label{lm:TSL05}
Under the hypotheses of Theorem~\ref{th05},
$T_{n,r,\kappa}^{(k)} \to \tfrac{\|p\|}{\lambda p_r}$ in probability, as $k\to\infty$.
\end{lemma}
\begin{proof}
For all $\epsilon>0$,
\begin{subequations}
\begin{align}
\Pr \Big(|T_{n,r,\kappa}^{(k)} - \tfrac{\|p\|}{\lambda p_r}| \ge \epsilon\Big) & = \Pr \Big(|T_{n,r,\kappa}^{(k)} - \E T_{n,r,\kappa}^{(k)}| \ge \epsilon\Big) \\
\label{TSL05_1}
&\le  \frac{1}{\epsilon^2} \Var\, T_{n,r,\kappa}^{(k)}\\
\label{TSL05_2}
&=  \frac{1}{\epsilon^2} \left(  \E (\Var\, T_{n,r,\kappa}^{(k)} | U_{r,\kappa}) +  \Var\, \E( T_{n,r,\kappa}^{(k)} | U_{r,\kappa})  \right) \\
\label{TSL05_3}
&=  \frac{1}{\epsilon^2} \left(  \frac{1}{p_r} \sum_{j=1}^{p_r} \Var\, A_{j,r,\kappa}^{(k)}  +   (\E A_{j,r,\kappa}^{(k)}-\E T_{n,r,\kappa}^{(k)})^2 \right) \\
\label{TSL05_4}
&=  \frac{1}{\epsilon^2} \left(  \frac{1}{p_r} \sum_{j=1}^{p_r} a_{j,r}^{(k)} \Var\, T_{1}^{(k)}  +   \Big(\tfrac{a_{j,r}^{(k)}}{\lambda k}-\tfrac{\|p\|}{\lambda p_r}\Big)^2 \right) \\
\label{TSL05_5}
& \xrightarrow[k\rightarrow\infty]{} 0.
\end{align}
\end{subequations}
In \eqref{TSL05_1}, \eqref{TSL05_2} and \eqref{TSL05_4}, we have used Chebyshev's inequality, the law of total variance and that the $T_{n}^{(k)}$'s are i.i.d., respectively.
In \eqref{TSL05_5}, we have used \eqref{eq:lim_a_jkr2} and that $a_{j,r}^{(k)} \Var\, T_{1}^{(k)}\to 0$ because $\Var\, T_{1}^{(k)}=o(k)$ and $a_{j,r}^{(k)}\le k \|p\|$.
\end{proof}

Since $T_{n,r,\kappa}^{(k)}$ converges in probability for each $n$, also the finite dimensional distributions of the process 
$(T_{n,r,\kappa}^{(k)})_{n\in\mathbb{N}}$ converge to the one of the constant process with rate $\tfrac{\lambda p_r}{\|p\|}$.
Together with the fact that $\E T_{n,r,\kappa}^{(k)}=\lim_{k\to\infty} \E T_{n,r,\kappa}^{(k)}=\tfrac{\|p\|}{\lambda p_r}$, we can use the continuity of the stationary waiting time (see \cite[Theorem~22]{borovkov1976stochastic}) to establish that
\begin{equation}
\label{eq:lim_Dist_Wk_pqsa_rk}
W_{r,\kappa}^{(k)}(\pi^{(k)}) \xrightarrow[k\rightarrow\infty]{d} W_{r} (p).
\end{equation}
Using \eqref{eq:lim_Dist_Wk_pqsa_rk} and that Cesaro sums converge if each addend converges, we obtain
\begin{subequations}
\begin{align}
\lim_{k\to\infty} \Pr(W^{(k)}(\pi^{(k)}) \le t)
& = \lim_{k\to\infty} \sum_{r=1}^R \frac{p_r}{\|p\|k} \sum_{\kappa=1}^k \Pr(W_{r,\kappa}^{(k)}(\pi^{(k)}) \le t)\\
& = \sum_{r=1}^R \frac{p_r}{\|p\|} \Pr(W_{r} (p) \le t)
=  \Pr(W (p) \le t)
\end{align}
as desired.
\end{subequations}

\subsection{Proof of Theorem~\ref{th1}}

\emph{Proof of~\eqref{eq:MLB} and~\eqref{eq:lim_E_Wk_pqsa}}.
Given that $\J_p^{(k)}\subseteq \overline{\A}_p(k)$, \eqref{eq:MLB} and~\eqref{eq:lim_E_Wk_pqsa} hold true if
\begin{equation}
\label{eq:dis1}
% \sum_r \frac{p_r}{\|p\|} \, \E W_r(p) \le
\E W (p) \le
\E W^{(k)}(\pi)
\end{equation}
for all $\pi\in \overline{\A}_p(k)$ and
\begin{equation}
\label{eq:dis2}
\lim_{k\to\infty} \E W^{(k)}(\pi^{(k)}) \le \E W (p).
\end{equation}
% where $(\pi^{(k)})_{k\in\mathbb{N}}$ is an arbitrary given sequence of policies where $\pi^{(k)}\in\J_p^{(k)}$ for all $k$.

Let $\E W_{r,\kappa}(x)$ be the mean waiting time of a D/GI/1 queue with arrival rate $\lambda x$ and i.i.d. service times having the same distribution of $S_{1,r,\kappa}$.
Inequality~\eqref{eq:dis1} is a fairly direct application of known results:
for all $\pi\in \overline{\A}_p(k)$,
\begin{subequations}
\label{eq:dis1_proved}
\begin{align}
\E W^{(k)}(\pi) &= \sum_{r,\kappa} q_{r,\kappa} \E W_{r,\kappa}^{(k)}(\pi)\\
\label{eq:fr_humblet} &\ge  \sum_{r,\kappa} q_{r,\kappa} \E W_{r,\kappa}(k q_{r,\kappa})\\
\label{eq:fr_convex} &\ge  \sum_{r,\kappa} \frac{p_r}{k\|p\|} \E W_{r,\kappa}(\tfrac{p_r}{\|p\|})\\
&=  \sum_r \frac{p_r}{\|p\|} \E W_{r}(p) =\E W (p).
\end{align}
\end{subequations}
In \eqref{eq:fr_humblet}, we have used the lower bound in \cite{humblet1982determinism}.
In \eqref{eq:fr_convex}, we have used Karamata's inequality once noticing that
i) $\E W_{r,\kappa}(x)=\E W_{r,1}(x)$,
ii) the majorization
% \footnote{See the appendix for a definition of the majorization relation $\prec$.}
$(\frac{p_r}{\|p\|},\ldots,\frac{p_r}{\|p\|}) \prec  (k q_{r,1},\ldots,k q_{r,k})$ holds, and
iii) the mean waiting time of a D/GI/1 queue is convex increasing in the arrival rate (see, e.g., \cite[Theorem~5]{NeelyM05}, \cite{gun90}), which means that $q_{r,\kappa} \E W_{r,\kappa}(k q_{r,\kappa})$ is convex in $q_{r,\kappa}$.

We now prove \eqref{eq:dis2}. As in the proof of Theorem~\ref{th05}, this can be done assuming that \eqref{eq:subseq_pik} holds.
Thus, the sequences \eqref{eq:subseq_pik} will be assumed along this proof. 

The remainder of the proof basically works as follows. 
First, we bound the waiting times of our G/GI/1 queues through the waiting times of suitable GI/GI/1 queues.
Second, we show that the sequence of such waiting times converges in distribution to $W(p)$.
Then, we show that the waiting times of such GI/GI/1 queues are non-increasing in the $\le_{icx}$-sense along the sequences $k_m+i$, for all $i=1,\ldots,\mbox{lcm}(p)$, which allows us to conclude that the sequence is uniformly integrable; this is the point where we will use Fact~\ref{f1}.
Finally, we use \cite[Theorem 3.5, pp. 31]{Bill} to conclude that also the sequence of the expected values converges to $\E W(p)$.

Associated to each queue of type~$r$, we define an auxiliary random variable, $\overline{T}_{r}^{(k)}$, such that
\begin{equation}
\label{eq:Tbar_min}
\overline{T}_{r}^{(k)}=_{st} \sum_{n=1}^{\min\limits_{j=1,\ldots,p_r} a_{j,r}^{(k)}} T_{n}^{(k)}.
\end{equation}

Next lemma provides properties satisfied by~$\overline{T}_{r}^{(k)}$ that will be used later.
We recall that $k_m=m \, \mbox{lcm}(p)$, see~\eqref{eq:k_m}.

\begin{lemma}
\label{lm:TSL}
Under the hypotheses of Theorem~\ref{th1}, the following properties hold:
\begin{itemize}
%  \item[i)] $\E T_{n,\kappa,r}^{(k)}=\tfrac{\|p\|}{p_r\lambda}$.
 \item[i)]
We have
\begin{equation}
\label{eq:lim_a_jkr}
\lim\limits_{k\rightarrow \infty} \E \overline{T}_{r}^{(k)} =
\lim\limits_{k\rightarrow \infty} \min_{j=1,\ldots,p_r}\tfrac{a_{j,r}^{(k)}}{\lambda k}=\tfrac{\|p\|}{\lambda p_r}.
\end{equation}

% \item[ii)]
% Let $m\in\mathbb{N}$ and let $k_m\bydef m\, \mbox{\emph{lcm}}(p)$, where \emph{lcm}$(p)$ denotes the least common multiple of $p_1,\ldots,p_R$. 
% Then, 
% $\E \overline{T}_{r}^{(k_m+i)}$
% is non-decreasing in $m$, for all $i=1,\ldots,\mbox{\emph{lcm}}(p)$.

% \item[ii)] $\overline{T}_{r}^{(k)} \xrightarrow[k\rightarrow\infty]{\Pr} \tfrac{\|p\|}{\lambda p_r}$.
\item[ii)] $\overline{T}_{r}^{(k)} \to \tfrac{\|p\|}{\lambda p_r}$ in probability, as $k\to\infty$.

\item[iii)] For all $i=1,\ldots,\mbox{\emph{lcm}}(p)$,
\begin{equation}
\label{eq:icx_min_T}
% -\min\limits_{j=1,\ldots,p_r} \sum_{n=1}^{a_{j,r}^{(k_{m+1}+i)}} T_{n}^{(k_{m+1}+i)} \le_{icx} -\min\limits_{j=1,\ldots,p_r} \sum_{n=1}^{a_{j,r}^{(k_{m}+i)}} T_{n}^{(k_{m}+i)} .
-\overline{T}_{r}^{(k_{m+1}+i)} \le_{icx} - \overline{T}_{r}^{(k_{m}+i)}.
\end{equation}
\end{itemize}
\end{lemma}
\begin{proof}
\emph{Proof of i)}.
This is an immediate consequence of \eqref{eq:lim_a_jkr2}.

\ \\
\emph{Proof of ii)}.
For all $i=1,\ldots,\mbox{lcm}(p)$,
let $j_i^*\in \arg\min_{j=1,\ldots,p_r} \frac{a_{j,r}^{(i)}}{i}$.
Then, from \eqref{eq:sim_ordered}, we get
\begin{equation}
\label{eq:argmin_dnc}
j_i^*\in \arg\min_{j=1,\ldots,p_r} \frac{a_{j,r}^{(k_{m}+i)}}{k_{m}+i}, \quad \forall m>1.
\end{equation}
% and  since $\frac{a_{j_i^*,r}^{(k_m+i)}}{k_m+i}\le \tfrac{\|p\|}{p_r}$ for all $m$ (by \eqref{eq:exp_Tn}),
% we obtain that $\frac{a_{j_i^*,r}^{(k_m+i)}}{\lambda (k_m+i)}$ is non-decreasing in $m$.

In view of Lemma~\ref{lm:ddard}, the convergence in ii) holds if we can show that
% \begin{equation}
% \min\limits_{j=1,\ldots,p_r} \sum_{n=1}^{a_{j,r}^{(k_m+i)}} T_{n}^{(k_m+i)}
$\overline{T}_{r}^{(k_m+i)} \xrightarrow[m\rightarrow\infty]{\Pr} \tfrac{\|p\|}{\lambda p_r}$,  for all $i=1,\ldots,\mbox{lcm}(p)$.
% \end{equation}
Given \eqref{eq:argmin_dnc}, this amounts to show that
\begin{equation}
\sum_{n=1}^{a_{j_i^*,r}^{(k_m+i)}} T_{n}^{(k_m+i)} \xrightarrow[m\rightarrow\infty]{\Pr} \tfrac{\|p\|}{\lambda p_r}, \quad \forall i=1,\ldots,\mbox{lcm}(p).
\end{equation}
We prove the former by showing (the stronger statement) that 
$\sum_{n=1}^{a_{j,r}^{(k)}} T_{n}^{(k)} \xrightarrow[k\rightarrow\infty]{\Pr} \tfrac{\|p\|}{\lambda p_r}$, for all $j$.
Now, using that
$\{| X-c | >2\epsilon\} 
\subseteq
\{| X-\E X |>\epsilon \} \cup
\{|\E X-c|>\epsilon \}$
for a random variable $X$, we have 
\begin{subequations}
\nonumber
\begin{align}
\Pr \Big(|\sum_{n=1}^{a_{j,r}^{(k)}} T_{n}^{(k)}-\tfrac{\|p\|}{p_r\lambda} | \ge \epsilon\Big) \le &
\Pr \Big(|\sum_{n=1}^{a_{j,r}^{(k)}} T_{n}^{(k)}-\E \sum_{n=1}^{a_{j,r}^{(k)}} T_{n}^{(k)} | \ge \epsilon\Big) + \\
&\Pr \Big(| \tfrac{\|p\|}{p_r\lambda} -\E \sum_{n=1}^{a_{j,r}^{(k)}} T_{n}^{(k)} | \ge \epsilon\Big). 
\end{align}
\end{subequations}
% \begin{equation}
% \Pr \Big(|\sum_{n=1}^{a_{j,r}^{(k)}} T_{n}^{(k)}-\tfrac{\|p\|}{p_r\lambda} | \ge \epsilon\Big) \le
% \Pr \Big(|\sum_{n=1}^{a_{j,r}^{(k)}} T_{n}^{(k)}-\E \sum_{n=1}^{a_{j,r}^{(k)}} T_{n}^{(k)} | \ge \epsilon\Big) + \Pr \Big(| \tfrac{\|p\|}{p_r\lambda} -\E \sum_{n=1}^{a_{j,r}^{(k)}} T_{n}^{(k)} | \ge \epsilon\Big).
% \end{equation}
The second term in the right-hand side of former inequality tends to zero as $k\to\infty$ by \eqref{eq:lim_a_jkr2}.
The following shows that also the first term goes to zero:
\begin{subequations}
\begin{align}
\label{eq:cP_cc1}
\Pr \Big(|\sum_{n=1}^{a_{j,r}^{(k)}} T_{n}^{(k)}-\E \sum_{n=1}^{a_{j,r}^{(k)}} T_{n}^{(k)} | \ge \epsilon\Big)
& \le \frac{1}{\epsilon^2} \sum_{n=1}^{a_{j,r}^{(k)}} \Var \, T_{n}^{(k)}  \\
% 
% & \le \frac{1}{\epsilon^2} \sum_{n=1}^{a_{j,r}^{(k)}} \Var \, T_{n}^{(k)}  \\
% 
\label{eq:cP_cc2}
% & \le \frac{a_{j,r}^{(k)}}{\lambda^2 k^2 \epsilon^2}
& \le \frac{1}{\epsilon^2} a_{j,r}^{(k)} o(1/k)
\xrightarrow[k\rightarrow\infty]{} 0.
\end{align}
\end{subequations}
In \eqref{eq:cP_cc1}, we have used Chebyshev's inequality and that the $T_{n}^{(k)}$'s are independent.
In \eqref{eq:cP_cc2}, we have used that
% $\Var \, T_{n}^{(k)}=c/k^2$ for some $c>0$ independent of $k$, which is a property of  the phase-type distribution,
$a_{j,r}^{(k)}\le k \|p\|$.
% Under Case~\ref{as2}, we have $\Var \, T_{n}^{(k)}=0$.

\ \\
\emph{Proof of iii)}.
We use that $X\le_{icx} Y$ if and only if there exists an other random variable $Z$ such that $X\le_{st} Z$ and $Z\le_{cx}Y$ \cite{AMM83}.

Using \eqref{eq:sim_ordered0} and that $\frac{a_{j_i^*,r}^{(k_{m}+i)}}{(k_{m}+i)}\le\tfrac{\|p\|}{p_r}$ for all $m$ (by \eqref{eq:sim_ordered}), the first observation is that
\begin{equation}
a_{j_i^*,r}^{(k_{m+1}+i)} \ge a_{j_i^*,r}^{(k_{m}+i)} +\mbox{lcm}(p)\frac{a_{j_i^*,r}^{(k_{m}+i)}}{k_{m}+i},
\end{equation}
where $j_i^*$ is defined above in point i).
Thus, we have
\begin{subequations}
\begin{align}
- \overline{T}_{r}^{(k_{m+1}+i)}&
=_{st} -  \sum_{n=1}^{a_{j_i^*,r}^{(k_{m+1}+i)}} T_{n}^{(k_{m+1}+i)} \\
&\le_{st} -  \sum_{n=1}^{a_{j_i^*,r}^{(k_{m}+i)} +\mbox{lcm}(p)\frac{a_{j_i^*,r}^{(k_{m}+i)}}{k_{m}+i}} T_{n}^{(k_{m+1}+i)}\bydef - Z.
\end{align}
\end{subequations}
Now, it remains to show that  $ - Z \le_{cx} - \overline{T}_{r}^{(k_{m}+i)}$, which is equivalent to show that
\begin{equation}
\label{Z_cx_T}
Z \le_{cx} \overline{T}_{r}^{(k_{m}+i)},
\end{equation}
see \cite[Theorem~3.A.12]{shantikumar}.
Since
\begin{equation}
\E Z = \frac{1}{\lambda} \frac{a_{j_i^*,r}^{(k_{m}+i)} +\mbox{lcm}(p)\frac{a_{j_i^*,r}^{(k_{m}+i)}}{k_{m}+i}}{k_{m}+i+\mbox{lcm}(p)}
= \frac{1}{\lambda} \frac{a_{j_i^*,r}^{(k_{m}+i)}}{k_{m}+i}
= \E \overline{T}_{r}^{(k_{m}+i)},
% \overline{T}_{r}^{(k_{m}+i)}.
\end{equation}
\eqref{Z_cx_T} holds trivially under Case~\ref{as2}.
Now, let Case~\ref{as1} hold.
Noticing that both $\overline{T}_{r}^{(k_{m}+i)}$ and $Z$ have Erlang distributions with the same mean,
to prove \eqref{Z_cx_T} is enough to show that $\Var\,Z \le \Var\, \overline{T}_{r}^{(k_{m}+i)}$; see \cite[p.~14]{stoyan1983comparison}.
We have
\begin{subequations}
\begin{align}
\lambda^2 \,\Var\,Z
& =  \frac{a_{j_i^*,r}^{(k_{m}+i)} +\mbox{lcm}(p)\frac{a_{j_i^*,r}^{(k_{m}+i)}}{k_{m}+i}}{(k_{m+1}+i)^2} \\
& =  \frac{a_{j_i^*,r}^{(k_{m}+i)}}{(k_m+i)^2} \frac{(k_m+i)(k_{m}+i + \mbox{lcm}(p))}{(k_{m}+i+\mbox{lcm}(p))^2} \\
& \le \frac{a_{j_i^*,r}^{(k_{m}+i)}}{(k_{m}+i)^2} = \lambda^2 \,\Var\, \overline{T}_{r}^{(k_{m}+i)}
\end{align}
\end{subequations}
as desired.
\end{proof}

% For simplicity, we refer to $T_{n,1,r}^{(k)}(\pi^{(k)})$ as $T_{n,1,r}^{(k)}$.

We now present an argument that allows us to uniformly bound the second moment of $W_{r,\kappa}^{(k)}$.

Let $\delta_r \bydef \tfrac{1}{2}\left( \tfrac{\|p\|}{p_r\lambda}-\tfrac{1}{\mu_r} \right)$ and
\begin{equation}
\label{eq:k_star}
k^*\bydef \min\left\{k>0: 0\le \tfrac{\|p\|}{p_r\lambda} - \min_{j=1,\ldots,p_r} \tfrac{a_{j,r}^{(k')}}{k'\lambda} \le \delta_r,\quad \forall k'\ge k \right\}.
\end{equation}

\begin{lemma}
\label{lm:kstar_inf}
$k^*<\infty$.
\end{lemma}
\begin{proof}
This is immediate because
$\delta_r>0$ by hypothesis, 
$\lim\limits_{k'\to\infty} \min\limits_{j=1,\ldots,p_r} \tfrac{a_{j,r}^{(k')}}{k'\lambda} = \tfrac{\|p\|}{p_r\lambda}$ (see \eqref{eq:lim_a_jkr}), and
$\tfrac{\|p\|}{p_r\lambda} \ge \min\limits_{j=1,\ldots,p_r} \tfrac{a_{j,r}^{(k)}}{k\lambda}$ for all $k$.
% most importantly $\min\limits_{j=1,\ldots,p_r} \tfrac{a_{j,r}^{(k_m+i)}}{(k_m+i)\lambda}$ is non-decreasing in $m$ (see \eqref{eq:sim_ordered}), for all $i=1,\ldots,\mbox{lcm}(p)$.
\end{proof}

% Furthermore, let
% \begin{equation}
% \label{eq:k_star2}
% % k^\star\bydef \min\left\{k\ge k^*: \min_{j=1,\ldots,p_r} \tfrac{a_{j,r}^{(k_m+i)}}{k_m+i}\ge \min_{j=1,\ldots,p_r} \tfrac{a_{j,r}^{(k_m+i)}}{k_m+i} ,\quad \forall k\ge k' \right\}.
% k^\star\bydef \min\left\{k\ge k^*: \min_{j=1,\ldots,p_r} \tfrac{a_{j,r}^{(k)}}{k}\le \min_{j=1,\ldots,p_r} \tfrac{a_{j,r}^{(k')}}{k'} ,\quad \forall k'\ge k \right\}.
% \end{equation}
% Since $\min\limits_{j=1,\ldots,p_r}\tfrac{a_{j,r}^{(k_m+i)}}{k_m+i}$ is non-decreasing in $m$, for all $i=1,\ldots,\mbox{lcm}(p)$ (by Lemma~\ref{lm:TSL}), $k^\star<\infty$.

Let $\overline{W}_r^{(k)}$ denote the stationary waiting time of a GI/GI/1 queue
with (i.i.d.) interarrival times $(\overline{T}_{n,r}^{(k)})_{n\in\mathbb{N}}$ where $\overline{T}_{n,r}^{(k)}=_{st} \overline{T}_{r}^{(k)}$ (see \eqref{eq:Tbar_min}) and service times $(S_{n,\kappa,r}^{(k)})_{n\in\mathbb{N}}$.
By coupling the $\overline{T}_{n,r}^{(k)}$'s and the $T_{n,\kappa,r}^{(k)}$'s in the obvious manner,
one can easily see that
$(\overline{T}_{1,r}^{(k)},\ldots,\overline{T}_{n,r}^{(k)}) \le (T_{1,\kappa,r}^{(k)},\ldots,T_{n,\kappa,r}^{(k)})$\footnote{Given $x,y\in\mathbb{R}^d$, here $x\le y$ means $x_i\le y_i$ for all $i=1,\ldots,d$.}
and therefore we have $(\overline{T}_{1,r}^{(k)},\ldots,\overline{T}_{n,r}^{(k)}) \le_{st} (T_{1,\kappa,r}^{(k)},\ldots,T_{n,\kappa,r}^{(k)})$.
% , where $\le_{st}$ denotes the usual stochastic order (e.g., \cite[p.~209]{baccelli2003elements}).
Using, e.g., \cite[pp.~217,~220]{baccelli2003elements}, this implies
\begin{equation}
\label{eq:Wst_copm}
W_{r,\kappa}^{(k)} \le_{st} \overline{W}_r^{(k)}.
\end{equation}
Furthermore, given that
$-\overline{T}_{n,r}^{(k_{m+1}+i)}\le_{icx} -\overline{T}_{n,r}^{(k_m+i)}$ for all $i=1,\ldots,\mbox{lcm}(p)$ (by Lemma~\ref{lm:TSL}) and that the $(\overline{T}_{n,r}^{(k)})_{n\in\mathbb{N}}$ are independent, we can use \cite[p.~337]{Asmussen} to establish that 
\begin{equation}
\label{eq:Wst_copm2}
\overline{W}_r^{(k_{m+1}+i)}\le_{icx} \overline{W}_r^{(k_{m}+i)},
\end{equation}
for all $i=1,\ldots,\mbox{lcm}(p)$.
Therefore, given $\E \left( (W_{r,\kappa}^{(k)})^2 \right)<\infty$ for all $k$ and Lemma \ref{lm:kstar_inf}, we can uniformly bound the second moment of $W_{r,\kappa}^{(k)}$ as follows
\begin{subequations}
\label{eq:UI}
\begin{align}
% \sup\limits_{k\ge k^*} \E \left( W_{r,\kappa}^{(k)}(\pi^{(k)})^2 \right)
% \sup\limits_{k\ge k^*} \E \left( \overline{W}_r\left( \min\limits_{j=1,\ldots,p_r} \tfrac{a_{j,r}^{(k)}}{k\lambda}\right)^2 \right)
\label{eq:UI1}
\sup\limits_{k\ge k^*} \E \left( (W_{r,\kappa}^{(k)})^2 \right)
& \le \sup\limits_{k\ge k^*} \E \left( (\overline{W}_r^{(k)})^2 \right)\\
\label{eq:UI2}
&= \max\limits_{i=1,\ldots,\mbox{lcm}(p)} \sup\limits_{m:k_m+i\ge k^*} \E \left( (\overline{W}_r^{(k_m+i)})^2 \right)\\
\label{eq:UI3}
&= \max\limits_{i=1,\ldots,\mbox{lcm}(p)} \E \left( (\overline{W}_r^{(k_{m_i^*}+i)})^2 \right) \\
\label{eq:UI4}
& <\infty,
\end{align}
\end{subequations}
where $m_i^*\bydef \min\{m:k_m+i\ge k^*\}$.
In \eqref{eq:UI1} and \eqref{eq:UI3}, we have used \eqref{eq:Wst_copm} and \eqref{eq:Wst_copm2}, respectively.
In \eqref{eq:UI4}, we have used that
\begin{equation}
% \E \min_{j=1,\ldots,p_r} \sum_{n=1}^{a_{j,r}^{(k)}} T_{n}^{(k)}
\E \overline{T}_{n,r}^{(k)}
= \min\limits_{j=1,\ldots,p_r} \tfrac{a_{j,r}^{(k)}}{k\lambda}
\ge \tfrac{\|p\|}{p_r\lambda} -\delta_r
= \tfrac{1}{2} \tfrac{\|p\|}{p_r\lambda}+\tfrac{1}{2}\tfrac{1}{\mu_r}
> \tfrac{1}{\mu_r}, \quad \forall k\ge k^*,
\end{equation}
i.e. the ergodicity condition, and that the third moment of service times is finite, which imply that the second moment of $\overline{W}_r^{(k)}$ is finite \cite[pg.~270]{Asmussen}.

% and together with the fact that $\E \left( W_r(\pi^{(k)})^2 \right)$ is finite for any $k$, we get
% \begin{equation}
% \label{eq:UI}
% \sup\limits_{k\ge k^*} \E \left( W_{r,\kappa}^{(k)}(\pi^{(k)})^2 \right)  <\infty.
% \end{equation}
Now, using the continuity of the stationary waiting time of GI/GI/1 queues \cite[Corollary~X.6.4]{Asmussen} and part i) and ii) of Lemma~\ref{lm:TSL}, we have
$
\overline{W}_r^{(k)}\xrightarrow[k\rightarrow\infty]{d}  W_r (p),
$
and given the uniform integrability \eqref{eq:UI} we have that also the expected values converge \cite[Theorem 3.5, pp. 31]{Bill}, i.e.,
\begin{equation}
\label{eq:exp_values_covnerge}
\lim_{k\to\infty} \E \overline{W}_r^{(k)} = \E W_r (p).
\end{equation}
With the above relations, we can conclude the proof of \eqref{eq:dis2}
\begin{subequations}
\begin{align}
% & \lim_{k\to\infty} \E W^{(k)}(\pi^{(k)})\\
\lim_{k\to\infty} \E W^{(k)}(\pi^{(k)})= & \lim_{k\to\infty} \sum_{r} \sum_{\kappa} \frac{p_r}{\|p\|k} \,\E W_{r,\kappa}^{(k)}(\pi^{(k)})\\
\le & \lim_{k\to\infty} \sum_{r} \frac{p_r}{\|p\|} \,\E \overline{W}_{r}^{(k)}\\
\label{eq:using_lemmacomparisonW}
= &\sum_r \frac{p_r}{\|p\|} \, \E W_r(p). 
\end{align}
\end{subequations}
% In~\eqref{eq:using_rr}, we have used that $\E W_{r,\kappa}^{(k)}(\pi^{(k)})=\E W_{r,1}^{(k)}(\pi^{(k)})$, which follows by symmetry because queues of the same type are visited in a round-robin order.
% In~\eqref{eq:using_lemmacomparisonW}, we have used \eqref{eq:exp_values_covnerge}.

\ \\
\emph{Proof of~\eqref{eq:lim_Var_Wk_pqsa}}.
Let $Q$ be a discrete random variable with values in $\{1,\ldots,R\}\times \{1,\ldots,k\}$ such that 
$\Pr(Q=(r,\kappa))=\tfrac{p_r}{\|p\| k}$. We assume that this random variable is independent of any other random variable.
By definition of $W^{(k)}(\pi^{(k)})$ and using the law of total variance, we obtain
\begin{subequations}
\label{eq:var_sss}
\begin{align}
\Var\, W^{(k)}(\pi^{(k)}) & = \E (\Var\, W^{(k)}(\pi^{(k)}) | Q) +  \Var\, \E( W^{(k)}(\pi^{(k)}) | Q) \\
& = \sum_{r,\kappa} \tfrac{p_r}{\|p\| k}   \left( \Var\, W_{r,\kappa}^{(k)}(\pi^{(k)}) +   \left( \E W_{r,\kappa}^{(k)}(\pi^{(k)}) - \E W^{(k)}(\pi^{(k)}) \right)^2 \right).
\end{align}
\end{subequations}
When $k\to\infty$, we have already established that
$\E W^{(k)}(\pi^{(k)})\to \E W(p)$
and that
$\E W_r(p)\le \E W_{r,\kappa}^{(k)}(\pi^{(k)})\le \E \overline{W}_{r,\kappa}^{(k)} \to \E W_r(p)$.
Therefore, it only remains to show that the second moment of $W_{r,\kappa}^{(k)}(\pi^{(k)})$ converge to $\E [W_r(p)^2]$.
This is done by using the same argument above for the convergence of the first moment.
% and the finiteness of the fifth moment of service times.
Hence, using the continuity of the waiting time and of the square function \cite[Corollary~X.6.4]{Asmussen} and part i) and ii) of Lemma~\ref{lm:TSL}, we obtain
$
(\overline{W}_r^{(k)})^2\xrightarrow[]{d}  W_r (p)^2
$, as $k\rightarrow\infty$.
Furthermore, the second moment of $(\overline{W}_r^{(k)})^2$ is finite because the fifth moment of the service times is finite \cite[pg.~270]{Asmussen} and
\eqref{eq:Wst_copm2} ensures that the sequence $(\overline{W}_r^{(k)})^2$ is uniformly integrable because it is non-increasing along subsequences $(k_m+i)_{m\in\mathbb{N}}$, for all $i=1,\ldots,\mbox{lcm}(p)$.
Thus, $\E[(\overline{W}_r^{(k)})^2] \to \E[W_r(p)^2]$.
Together with \eqref{eq:icx_W_r},
% \cite[p.~337]{Asmussen},
as desired we obtain
\begin{equation}
\lim_{k\to\infty} \E [(W_{r,\kappa}^{(k)}(\pi^{(k)}))^2] =  \E[W_r(p)^2].
\end{equation}

\section{Conclusions}
\label{sec:concl}

We have derived structural properties concerning a known problem in the literature of stochastic scheduling, that is Problem~\ref{prob1}.
Fixing the proportion of jobs to send on each queue, $p$,
we have identified a class of periodic policies and have proven that all the policies in this class are asymptotically equivalent and optimal.
The limiting mean waiting time achieved by these policies, $\E W(p)$ (see \eqref{eq:W_p_ascs}), is expressed in terms of a linear combination of independent D/GI/1 queues and
has the convenient property of being convex in~$p$.
We believe that these structural properties provide researchers and practitioners with new means
about the considered problem.
For instance, one consequence of these results is that the problem of computing the optimal proportions of jobs to send to each queue, which is considered a difficult problem (see the introduction), boils down, asymptotically, to the solution of an optimization problem of the form:
$$\min \E W(p)\quad \mbox{s.t.:} \quad p\in\mathcal{S},$$ for $\mathcal{S}$ compact and convex,
and we stress that $\E W(p)$ is a convex function of~$p$.
Using a classic result in convex optimization, this means that a polynomial number of evaluations of the objective function $\E W(p)$ are sufficient to converge to an optimizer of the problem.
Given that each objective evaluation is efficient \cite{NeelyM05,SX97,Squillante1999,Borst1995,ComBox95,BS83}, this lets us conclude that we have significantly reduced much of the difficulty of Problem~\ref{prob1}.
In the case where service times have an exponential distribution, $\E W(p)$ admits a very simple characterization because it is the weighted mean waiting time of $R$ D/M/1 queues~\cite{bhat,Anselmi2011}.

% As discussed in Section~\ref{sec:per_pol}, the policies in our set $\J_p^{(k)}$ can be implemented in a distributed manner in a two-tier hierarchy.
% Here, our results imply that the dispatcher in the first tier has a negligible influence on jobs waiting times, which means that it does not need to have full information about the statistics of network resources.

% Finally, to observe how fast the expected value $\E W^{(k)}$ converge to its asymptotic value $\E W(p)$, we have performed trace-driven simulations with data taken from real measurements of web servers and publicly available at the Internet Traffic Archive.
% It turns out that the mean waiting times achieved by both of them are very close each other even when $k$ is small, i.e., 10.
% For a complete description of these simuilation results, we point the interested reader to \cite{AGN13_arxiv}.

\bibliographystyle{abbrv}

% \appendix
% \input{SSY/appendix}

\end{document}